\documentclass{eptcs}
\providecommand{\event}{ICE 2016} 

\usepackage[dvipsnames,usenames]{xcolor}
\usepackage{amsmath,amssymb,dsfont,wasysym,stmaryrd,mathrsfs,turnstile,cancel,graphicx,caption,subcaption,wrapfig,hyperref}
\newdimen\proofrulebreadth \proofrulebreadth=.05em
\newdimen\proofdotseparation \proofdotseparation=1.25ex
\newdimen\proofrulebaseline \proofrulebaseline=2ex
\newcount\proofdotnumber \proofdotnumber=3
\let\then\relax
\def\hfi{\hskip0pt plus.0001fil} 
\mathchardef\squigto="3A3B
\newif\ifinsideprooftree\insideprooftreefalse
\newif\ifonleftofproofrule\onleftofproofrulefalse
\newif\ifproofdots\proofdotsfalse
\newif\ifdoubleproof\doubleprooffalse
\let\wereinproofbit\relax
\newdimen\shortenproofleft
\newdimen\shortenproofright
\newdimen\proofbelowshift
\newbox\proofabove
\newbox\proofbelow
\newbox\proofrulename
\def\shiftproofbelow{\let\next\relax\afterassignment\setshiftproofbelow\dimen0 }
\def\shiftproofbelowneg{\def\next{\multiply\dimen0 by-1 }%
\afterassignment\setshiftproofbelow\dimen0 }
\def\setshiftproofbelow{\next\proofbelowshift=\dimen0 }
\def\setproofrulebreadth{\proofrulebreadth}

\def\prooftree{
\ifnum  \lastpenalty=1
\then   \unpenalty
\else   \onleftofproofrulefalse
\fi
\ifonleftofproofrule
\else   \ifinsideprooftree
        \then   \hskip.5em plus1fil
        \fi
\fi
\bgroup
\setbox\proofbelow=\hbox{}\setbox\proofrulename=\hbox{}%
\let\justifies\proofover\let\leadsto\proofoverdots\let\Justifies\proofoverdbl
\let\using\proofusing\let\[\prooftree
\ifinsideprooftree\let\]\endprooftree\fi
\proofdotsfalse\doubleprooffalse
\let\thickness\setproofrulebreadth
\let\shiftright\shiftproofbelow \let\shift\shiftproofbelow
\let\shiftleft\shiftproofbelowneg
\let\ifwasinsideprooftree\ifinsideprooftree
\insideprooftreetrue
\setbox\proofabove=\hbox\bgroup$\displaystyle 
\let\wereinproofbit\prooftree
\shortenproofleft=0pt \shortenproofright=0pt \proofbelowshift=0pt
\onleftofproofruletrue\penalty1
}

\def\eproofbit{
\ifx    \wereinproofbit\prooftree
\then   \ifcase \lastpenalty
        \then   \shortenproofright=0pt  
        \or     \unpenalty\hfil         
        \or     \unpenalty\unskip       
        \else   \shortenproofright=0pt  
        \fi
\fi
\global\dimen0=\shortenproofleft
\global\dimen1=\shortenproofright
\global\dimen2=\proofrulebreadth
\global\dimen3=\proofbelowshift
\global\dimen4=\proofdotseparation
\global\count255=\proofdotnumber
$\egroup  
\shortenproofleft=\dimen0
\shortenproofright=\dimen1
\proofrulebreadth=\dimen2
\proofbelowshift=\dimen3
\proofdotseparation=\dimen4
\proofdotnumber=\count255
}

\def\proofover{
\eproofbit 
\setbox\proofbelow=\hbox\bgroup 
\let\wereinproofbit\proofover
$\displaystyle
}%
\def\proofoverdbl{
\eproofbit 
\doubleprooftrue
\setbox\proofbelow=\hbox\bgroup 
\let\wereinproofbit\proofoverdbl
$\displaystyle
}%
\def\proofoverdots{
\eproofbit 
\proofdotstrue
\setbox\proofbelow=\hbox\bgroup 
\let\wereinproofbit\proofoverdots
$\displaystyle
}%
\def\proofusing{
\eproofbit 
\setbox\proofrulename=\hbox\bgroup 
\let\wereinproofbit\proofusing
\kern0.3em$
}

\def\endprooftree{
\eproofbit 
  \dimen5 =0pt
\dimen0=\wd\proofabove \advance\dimen0-\shortenproofleft
\advance\dimen0-\shortenproofright
\dimen1=.5\dimen0 \advance\dimen1-.5\wd\proofbelow
\dimen4=\dimen1
\advance\dimen1\proofbelowshift \advance\dimen4-\proofbelowshift
\ifdim  \dimen1<0pt
\then   \advance\shortenproofleft\dimen1
        \advance\dimen0-\dimen1
        \dimen1=0pt
        \ifdim  \shortenproofleft<0pt
        \then   \setbox\proofabove=\hbox{%
                        \kern-\shortenproofleft\unhbox\proofabove}%
                \shortenproofleft=0pt
        \fi
\fi
\ifdim  \dimen4<0pt
\then   \advance\shortenproofright\dimen4
        \advance\dimen0-\dimen4
        \dimen4=0pt
\fi
\ifdim  \shortenproofright<\wd\proofrulename
\then   \shortenproofright=\wd\proofrulename
\fi
\dimen2=\shortenproofleft \advance\dimen2 by\dimen1
\dimen3=\shortenproofright\advance\dimen3 by\dimen4
\ifproofdots
\then
        \dimen6=\shortenproofleft \advance\dimen6 .5\dimen0
        \setbox1=\vbox to\proofdotseparation{\vss\hbox{$\cdot$}\vss}%
        \setbox0=\hbox{%
                \advance\dimen6-.5\wd1
                \kern\dimen6
                $\vcenter to\proofdotnumber\proofdotseparation
                        {\leaders\box1\vfill}$%
                \unhbox\proofrulename}%
\else   \dimen6=\fontdimen22\the\textfont2 
        \dimen7=\dimen6
        \advance\dimen6by.5\proofrulebreadth
        \advance\dimen7by-.5\proofrulebreadth
        \setbox0=\hbox{%
                \kern\shortenproofleft
                \ifdoubleproof
                \then   \hbox to\dimen0{%
                        $\mathsurround0pt\mathord=\mkern-6mu%
                        \cleaders\hbox{$\mkern-2mu=\mkern-2mu$}\hfill
                        \mkern-6mu\mathord=$}%
                \else   \vrule height\dimen6 depth-\dimen7 width\dimen0
                \fi
                \unhbox\proofrulename}%
        \ht0=\dimen6 \dp0=-\dimen7
\fi
\let\doll\relax
\ifwasinsideprooftree
\then   \let\VBOX\vbox
\else   \ifmmode\else$\let\doll=$\fi
        \let\VBOX\vcenter
\fi
\VBOX   {\baselineskip\proofrulebaseline \lineskip.2ex
        \expandafter\lineskiplimit\ifproofdots0ex\else-0.6ex\fi
        \hbox   spread\dimen5   {\hfi\unhbox\proofabove\hfi}%
        \hbox{\box0}
        \hbox   {\kern\dimen2 \box\proofbelow}}\doll%
\global\dimen2=\dimen2
\global\dimen3=\dimen3
\egroup 
\ifonleftofproofrule
\then   \shortenproofleft=\dimen2
\fi
\shortenproofright=\dimen3
\onleftofproofrulefalse
\ifinsideprooftree
\then   \hskip.5em plus 1fil \penalty2
\fi
}

\newtheorem{theorem}{Theorem}[section]
\newtheorem{corollary}[theorem]{Corollary}
\newtheorem{lemma}[theorem]{Lemma}
\newtheorem{proposition}[theorem]{Proposition}
\newtheorem{definition}[theorem]{Definition}

\newtheorem{exmp}[theorem]{Example}

\newenvironment{proof}{\vspace{-0.1em}\noindent\textit{Proof}.\ }{\hfill$\blacksquare$}

\definecolor{red}{rgb}{1,0,0}

\definecolor{purple}{rgb}{.5,.0,.5}

\DeclareMathAlphabet{\mathpzc}{OT1}{pzc}{m}{it}

\usepackage{listings}

\newcommand{\IsDefAs}{\ensuremath{\stackrel{\textnormal{\tiny{def}}}{=}}}
\newcommand{\InitSays}{\sststile{}{\circ}}
\newcommand{\RuleFont}[1]{(\textsc{#1})}
\newcommand{\NotUnknownTo}{\ensuremath{\not\hspace{0.2em}?}}
\newcommand{\Wildcard}{\ensuremath{{}_{\textnormal{\textemdash}}}}
\newcommand{\Supposing}{\ensuremath{\mathbin{+_?}}}
\newcommand{\SuppUp}[1]{\ensuremath{[#1]_{_?}}}
\newcommand{\NotRel}{\ensuremath{\not\hspace{-0.1em}r}}
\newcommand{\NotConc}{\ensuremath{\not\hspace{0.1em}\parallel}}
\newcommand{\SingleC}{\ensuremath{\stackrel{c}{\rightarrow}}}
\newcommand{\SingleCNamed}[1]{\ensuremath{\stackrel{#1}{\rightarrow}}}
\newcommand{\MultiCs}{\ensuremath{\stackrel{\overline{c}}{\twoheadrightarrow}}}
\newcommand{\MultiCsPed}{\ensuremath{\stackrel{\overline{c}'}{\twoheadrightarrow}}}
\newcommand{\MultiCsNamed}[1]{\ensuremath{\stackrel{#1}{\twoheadrightarrow}}}
\newcommand{\ForwardTS}{\ensuremath{\mathscr{T}_F(W)}}
\newcommand{\ForBiSim}{\ensuremath{\sim_F}}
\newcommand{\SingleBC}{\ensuremath{\stackrel{c}{\leftarrow}}}
\newcommand{\BackwardTS}{\ensuremath{\mathscr{T}_B(W)}}
\newcommand{\BackBiSim}{\ensuremath{\sim_B}}
\newcommand{\LastRuleOf}[1]{\ensuremath{\mathsf{lr}(#1)}}
\newcommand{\MultiBCs}{\ensuremath{\stackrel{\overline{c}}{\twoheadleftarrow}}}
\newcommand{\SingleBCNamed}[1]{\ensuremath{\stackrel{#1}{\leftarrow}}}
\newcommand{\MultiBCsPed}{\ensuremath{\stackrel{\overline{c}'}{\twoheadleftarrow}}}

\makeatletter
\newcommand{\labitem}[2]{%
\def\@itemlabel{(\textbf{#1})}
\item
\def\@currentlabel{#1}\label{#2}}
\makeatother

\begin{document}

\title{Worlds of Events\\{\Large Deduction with Partial Knowledge about Causality}}

\author{Seyed Hossein Haeri \qquad Peter Van Roy
\institute{Universit\'e catholique de Louvain, Belgium}
\and
Carlos Baquero
\institute{Universidade do Minho, Portugal}
\and
Christopher Meiklejohn
\institute{Universit\'e catholique de Louvain, Belgium}
}

\def\titlerunning{Worlds of Events}
\def\authorrunning{S. H. Haeri, P. Van Roy, C. Baquero, \& C. Meiklejohn}

\maketitle

\begin{abstract}
Interactions between internet users are mediated by their devices and the common support infrastructure in data centres.
Keeping track of causality amongst actions that take place in this distributed system is key to provide a seamless interaction where effects follow causes.
Tracking causality in large scale interactions is difficult due to the cost of keeping large quantities of metadata;
even more challenging when dealing with resource-limited devices.
In this paper, we focus on keeping partial knowledge on causality and address deduction from that knowledge. 

We provide the first proof-theoretic causality modelling for distributed partial knowledge.
We prove computability and consistency results.
We also prove that the partial knowledge gives rise to a weaker model than classical causality.
We provide rules for offline deduction about causality and refute some related folklore.
We define two notions of forward and backward bisimilarity between devices, using which we prove two important results.
Namely, no matter the order of addition/removal, two devices deduce similarly about causality so long as: (1) the same causal information is fed to both. (2) they start bisimilar and erase the same causal information.
Thanks to our establishment of forward and backward bisimilarity, respectively, proofs of the latter two results work by simple induction on length.
\end{abstract}

\section{Introduction}

Causality \cite{Lamp:1978,Schw+Matt:1994} is an essential for our perception of the physical world, and of our relations to other entities.
If one puts a cup on a table, and looks back at it, one expects it to be there.
One also expects to get a reply to one's postcards, \textbf{after} they were sent, and not before.  

Given the fault-tolerance and high availability expected of internet-based services today, distributed algorithms have become ubiquitous.
One duty of these algorithms is to order events totally across multiple replicas of a service.
This total order is required to ensure computation determinism; given the requirement of having multiple replicas appear as a single system, each replica must implement a state machine which observes the same events in the same order~\cite{Schn:1990}.
However, because of the amount of coordination required, a total order in the entire distributed system is not always feasible while maintaining availability~\cite{Gilb+Lync:2002}.

Given the intractability of a total order, techniques that favour a partial order based on causality are explored for they express user's \textbf{intent}.
For a key-value store, one's writes may, e.g., be directed to one replica, and, subsequent reads served from a replica which has not yet received those writes.
If we consider the canonical example of an access control list for viewing photos~\cite{Lloy+Free+Kami+Ande:2011}, one would expect that a write operation removing Eve from having access to Alice's photos prior to Alice uploading a photo she did not want Eve to see, would be observed in this order by Eve when performing read operations.
 
However, tracking causality can be very expensive in terms of metadata size; more so when interactions amongst many distinct entities are targeted \cite{Char:1991}.
Devising scalable solutions to causality tracking is a demanding problem \cite{Lloy+Free+Kami+Ande:2011,Serg+Baqu+Gonc+Preg+Font:2014} to the extent that some solutions even accept to lose causal information by pruning metadata \cite{Deca+Hast+Jamp+Kaku+Laks+Pilc+Siva+Voss+Voge:2007}.
Be it because available resources are limited (say to an edge device) or because a replica (say in a data centre) is temporarily out-of-sync, only a partial view of the system causality might be available.
There is not much study on dealing with that partiality of knowledge, however.
In this paper, we address that problem via a proof-theoretic modelling for partial causality knowledge of distributed systems.
Partiality is not a loss in our model:
What is not stored might be deducible -- acting in favour of metadata size reduction.

Contributions of this paper are as follows:
\begin{enumerate}
  \item We model distributed causality such that the holistic system and the partial causal knowledge of a device are categorically distinct (Definitions~\ref{Defn:WoE} and \ref{Defn:MCosm}).
  \item We offer rules for deducing causality when a device is online (Definition~\ref{Defn:Micro.Decision}) and prove its computability (Theorem~\ref{Theo:Rules.Comp}) and consistency (Theorem~\ref{Theo:Micro.Rules.Const}).
  \item We show that deduction of causality with partial knowledge is strictly less accurate than the holistic causal knowledge (Lemma~\ref{Lemm:WoE.More.Accurate.Micro}) and that the deductions of different devices do not conflict (Corollary~\ref{Corr:Mu.Hats}).
  \item We offer rules for a device to deduce causal information independent of new causal data from outside, e.g., when offline (Definition~\ref{Defn:Offline.Rules}) and prove its consistency with the online rules (Lemmata~\ref{Lemm:Offline.Not.Refute.Online} and \ref{Lemm:Offline.Not.Disagree.Online}).
    We also prove a related folklore wrong (Lemma~\ref{Lemm:Sample.Refutation.1}) using the latter machinery.
  \item We craft a notion of bisimilarity (Definition~\ref{Defn:Micro.Analog}) and prove that the order of arrival of new causal data is insignificant for bisimilar devices (Theorem~\ref{Theo:For.Order.Irrelevant}).
  \item We craft another notion of bisimilarity (Definition~\ref{Defn:Back.Bisimilarity}) to prove that the order of removal of causal data is also insignificant for bisimilar devices (Theorem~\ref{Theo:Back.Order.Irrelevant}).
\end{enumerate}

Unlike traditional approaches to causality modelling that store a partial order of known causally related events and consider non related events as concurrent events, we explicitly model concurrency information and provide a broader spectrum of relations amongst events.  

\paragraph*{Real-World Benefits} \label{Sect:Benefits}

The technical developments of this paper are beneficial in the following ways:

Firstly, whilst being more general, our forward bisimilarity (Definition~\ref{Defn:For.Bisimilarity}) captures replication: like-stated replicas are bisimilar.
Replication in distributed systems serves fault-tolerance in that, for example, a like-stated replica will cover for a crashed replica.
The idea is that, because the replica providing the cover was like-stated, the crash will go \textbf{unobserved}.
Our forward bisimilarity serves that by its formalism for observational equivalence.
Secondly, offline decision making (Definition~\ref{Defn:Offline.Rules}) entails that, in presence of network partitions, a device gone offline will still be able to make (useful) new deductions (e.g., Lemma~\ref{Lemm:Sample.Refutation.1}).
It only is that the new deductions may not be at the same level of accuracy as those of its bisimilar devices that are still connected or when the device itself retrieves connection (Lemma~\ref{Lemm:WoE.More.Accurate.Micro}).
That is the service offline decision making provides to fault-tolerance.
Thirdly, Theorems~\ref{Theo:For.Order.Irrelevant} and \ref{Theo:Back.Order.Irrelevant} are formal characterisations for strong eventual consistency---so long as all the correspondences arrive/leave, the replicas are causally consistent, i.e., forward/backward bisimilar---serving key-value stores.

\section{Worlds of Events and Microcosms}

Call a binary relation $R$ a strict partial order on a set $P$ when $R$ is irreflexive, asymmetric, and transitive.
Then, we say that $(P, R)$ is a strict poset.
For a strict poset $(T, R)$, when $R$ is also total, call $(T, R)$ a strict chain.
Let $R$ be a relation on a set $S$.
For a subset $U$ of $S$, the symbol $R_{|U}$ denotes $R$ restricted to $U$.
We use ``$\veebar$" for the exclusive or of mathematical logic.
For a set $S$, write $|S|$ for the cardinality of $S$.
As is common in Set Theory, $\aleph_0$ denotes the cardinality of Natural Numbers ($\mathbb{N}$).
Throughout this paper, ``\Wildcard" is our wild card; its usage expresses our lack of interest in the exact details of what ``\Wildcard" has replaced.

\begin{definition} \label{Defn:WoE}
  Call $W(<, \parallel)$ a world of events when:
  \begin{enumerate}
    \labitem{W1}{Misc:W1} $W$ is an infinitely countable set (i.e., $|W| = \aleph_0$) of events that are ranged over by $e_1, e_2, \dots, e, e', \dots$,
    \labitem{W2}{Misc:W2} $<$ and $\parallel$ are binary relations defined on $W$ that are ranged over by $r_1, r_2, \dots,$ $r, r', \dots$,
    \labitem{W3}{Misc:W3} $(W, <)$ is a strict poset,
    \labitem{W4}{Misc:W4} $\parallel$ is irreflexive and symmetric, and
    \labitem{W5}{Misc:W5} $e_1 \neq e_2$ iff $e_1 \parallel e_2 \veebar e_1 < e_2 \veebar e_2 < e_1$.
  \end{enumerate}
  For $e_1, e_2 \in W$, when $(e_1, e_2) \in r$ for $r \in \{<, \parallel\}$, we write $W \vDash e_1\ r\ e_2$ and say $e_1\ r\ e_2$ holds for $W$.$\Square$
\end{definition}

The relations $<$ and $\parallel$ denote the familiar \textsf{happens-before} and \textsf{is-concurrent-with}, respectively \cite{Lamp:1978}.
Notice that here we define $\parallel$ explicitly \--- whilst the usual derived definition for non strict posets covers the elements that are not related in the order by stating $e_1 \parallel e_2$ iff $e_1 \not\leq e_2 \wedge e_2 \not\leq e_1$. 
Notice also that, in line with the traditional understanding about it \cite[Observation 1.3]{Schw+Matt:1994}, (\ref{Misc:W4}) does not define $\parallel$ transitive.

For a world of events, we take the relation $<>$ (read \textsf{is-causally-related-to})%
\footnote{For a use of $<>$ in reality, see the CISE proof system \cite{Gots+Yang+Ferr+Naja+Shap:2016}.
In a valid CISE execution, when a pair of events $e$ and $e'$ possess conflicting tokens, it is required that $e <> e'$.}
as a syntactic sugar for $< \cup <^{-1}$, namely, $e_1 <> e_2 \IsDefAs e_1 < e_2 \vee e_2 < e_1$.
Hence, $<>$ is symmetric. We can also observe that every distinct pair of events are attributed to exactly one of the basic relations, and that ``$<$'' $\cap$ ``$<^{-1}$'' $\cap$ ``$=$'' $\cap$ ``$\parallel$'' is always $\emptyset$. 

Fix the set of \textbf{accurate} relations $R = \{<, \parallel\}$.
The relation $<>$ is an \textbf{inaccurate} relation in that it does not expose the exact known direction of $<$.
We now extend $\vDash$ to $\vDash^*$ for when the inaccurate relation $<>$ is also needed to be taken into consideration.
Write $W \vDash^* e_1\ r\ e_2$ iff: $W \vDash e_1\ r\ e_2$; or, $r =\ <>$ and either $W \vDash e_1 < e_2$ or $W \vDash e_2 < e_1$.
Note that, unlike $\vDash$, not every distinct pair of events are attributed to a unique relation by $\vDash^*$.
In particular, for every $e_1$ and $e_2$ such that $W \vDash e_1 < e_2$, by definition, it is both the case that $W \vDash^* e_1 < e_2$ and $W \vDash^* e_1 <> e_2$.
We call $e_1\ r\ e_2$ a correspondence, ranged over by $c_1, c_2, \dots, c, c', \dots$
For a world of events $W$, we also fix $\mathcal{C}^*_W = \{c \mid W \vDash^* c\}$.
For $c = e_1\ r\ e_2$, we say $c$ is an accurate correspondence when $r$ is accurate.
We call $c$ inaccurate otherwise.

\begin{proposition} \label{Prop:WoE.Const}
  Every world of events $W$ is consistent: $W \vDash e_1\ r\ e_2$ and $W \vDash e_1\ r'\ e_2$ imply $r = r'$.
\end{proposition}

\begin{definition} \label{Defn:MCosm}
  Let $W(<, \parallel)$ be a world of events.
  Call $M(I, E)$ a microcosm of $W$ (write $M \vartriangleleft W$) when:
  \begin{enumerate}
    \labitem{M1}{Misc:M1} $I \subset W$ and $|I| < \aleph_0$,
    \labitem{M2}{Misc:M2} $(I, <_{|I})$ is a strict chain, 
    \labitem{M3}{Misc:M3} $E \subset \mathcal{C}^*_W$ and $|E| < \aleph_0$,
    \labitem{M4}{Misc:M4} $e_1\ r\ e_2 \in E$ implies that there is no chain of events $e'_1, \dots, e'_n$ in $M$ such that $e_1 = e'_1 <_{|M} \dots <_{|M} e'_n = e_2$ or $e_2 = e'_1 <_{|M} \dots <_{|M} e'_n = e_1$.%
      \footnote{More on the motivation behind (\ref{Misc:M4}) to follow.}
  \end{enumerate}
  Accordingly, call $W$ the enclosing world of $M$ and let $\mathcal{M}_W = \{M \mid M \vartriangleleft W \}$.$\Square$
\end{definition}

The difference between the notation we use for worlds of events and the one we use for microcosms might cause confusion at the first glance.
In addition to being the world of events, the $W$ in $W(<, \parallel)$ is a set, $<$ and $\parallel$ are relations on which.
To the contrary, the $M$ in $M(I, E)$ is only a name for the pair $(I, E)$.
Furthermore, $I$ is a set of events, whilst $E$ is a set of correspondences; they are not of the same sort.

For an $M(I, E)$, we refer to $I$ as the \textbf{internal} events of $M$, and, to $E$ as the set of \textbf{external} correspondences known to it.
When appropriate, we use the alternative notions $I(M)$ and $E(M)$, respectively.
Write $e_1 < e_2 \in I$ when $e_1, e_2 \in I$ and $W \vDash e_1 < e_2$.
Besides, write $e_1\ r\ e_2 \in M$, when $e_1\ r\ e_2 \in I$ or $e_1\ r\ e_2 \in E$.
Write $e \in E$ when $\exists e' \in W.\ e\ \Wildcard\ e' \in E \vee e'\ \Wildcard\ e \in E$.
Finally, write $e \in M$ when $e \in I$ or $e \in E$.
That is how we formalise the notion of microcosm membership informally used in (\ref{Misc:M4}).

A microcosm is our abstraction for a single state -- out of the possibly many states -- of a generic device.
The enclosing world of events of a microcosm is the abstraction we use for the ecosystem in which a device lives.
Certain events can take place locally for a device; in which case, they are stored in the internal events of the respective microcosm.
The correspondence between certain events can also be disclosed to a device by the ecosystem; in which case, they are stored in the external correspondences of the respective microcosm.
In our model, devices do not get to communicate directly with one another.
The ecosystem sits between devices in that news from other devices in the same ecosystem arrives via the ecosystem (as opposed to the other devices themselves).

Note that, unlike a world of events, for a microcosm, the relation $<>$ is not a syntactic sugar.
To the latter, an $<>$ instance is all the information that is given for the respective pair of events.
In that case, whilst no stronger information about the given pair is provided to the microcosm, the enclosing world of events is aware of the exact $<$ direction between the pair.
It, nevertheless, follows from irreflexiveness of $<$ that $<>$ is irreflexive too -- both for worlds of events and their microcosms.

\begin{exmp}
  For a microcosm $M$ such that $I(M) = e_1 < e_2 < e_3$, no correspondence $e_1\ r\ e_2$ can exist in $E(M)$ or (\ref{Misc:M4}) will be violated.
  There is no need for any ``order" to exist between all the events a microcosm knows of -- be it partial or total.
  $M' = (\varnothing, \{e_1 <> e_2, e_2 <> e_3\})$ is an entirely fine microcosm (perhaps of the same world of events as $M$), in which there is neither a total order nor a partial order between $e_1$, $e_2$, and $e_3$.
  $M'' = (M' + e3 < e4)$\footnote{Notation defined at the end of section.} is another permissible microcosm -- regardless of whether $e_3 < e_4 \in I(M'')$ or $e_3 < e_4 \in E(M'')$.
  Note that $M''$ does contain a partial order but no total one.
  Finally, $M''' = (I(M), \varnothing)$ is yet another valid microcosm, in which there truly is a total order.
\end{exmp}

\begin{wrapfigure}{r}{0.39\textwidth}
  \vspace{-2.5em}
  \begin{center}
  \hrule
    \vspace{-1em}
    $$
      \framebox[1.1\width]{$M \InitSays e_1\ r\ e_2$} \qquad \textnormal{where } r \in R \cup \{<>\}
    $$
    $$
      \begin{prooftree}
          e_1\ r\ e_2 \in M
        \justifies
          M \InitSays e_1\ r\ e_2
        \using
          \RuleFont{Init}
      \end{prooftree}
    $$
    $$
      \begin{prooftree}
          M \InitSays e_1 < e_2 \quad M \InitSays e_2 < e_3
        \justifies
          M \InitSays e_1 < e_3
        \using
          \RuleFont{In-Tr}
      \end{prooftree}
    $$
    $$
      \begin{prooftree}
          M \InitSays e_1 \parallel e_2
        \justifies
          M \InitSays e_2 \parallel e_1
        \using
          \RuleFont{Co-Sym}
      \end{prooftree}
    $$
    $$
      \begin{prooftree}
          M \InitSays e_1 <> e_2
        \justifies
          M \InitSays e_2 <> e_1
        \using
          \RuleFont{CR-Sym}
      \end{prooftree}
    $$
    \vspace{-0.5em}
  \hrule
  \end{center}
  \vspace{-0.5em}
  \caption{Microcosm Initial Judgements}
  \label{Fig:Init.Rules}
  \vspace{-6em}
\end{wrapfigure}

We now introduce the first sort of deduction for microcosms (Definition~\ref{Defn:Init.Judge}).
The idea is that such a deduction is for a microcosm to decree on the correspondences it does know of.
Later in Section~\ref{Sect:Online.Decision}, we will generalise deduction for a microcosm to also conclude that it does not know the correspondence between a given pair of events.

\begin{definition} \label{Defn:Init.Judge}
  Let $M$ be a microcosm.
  Judgements of the form $M \InitSays e_1\ r\ e_2$ are called the initial judgements of $M$ when they are derived using the rules in Fig.~\ref{Fig:Init.Rules}.
  Write $M \not\InitSays e_1\ r\ e_2$ when $M \InitSays e_1\ r\ e_2$ is not true.$\Square$
\end{definition}

Note that with ``\Wildcard'' being existential in nature, the negation acts universally.
In particular, $M \not\InitSays e_1\ \Wildcard\ e_2$ stipulates the lack of \textit{any} initial correspondence between $e_1$ and $e_2$ in $M$.

Here is an informal account of the rules in Fig.~\ref{Fig:Init.Rules}:
\RuleFont{Init} states that every piece of information that is initially provided to a microcosm is reliable in the initial judgements made inside that microcosm.
\RuleFont{In-Tr} legislates transitivity of $<$ for initial judgements (regardless of whether the premises come from internal or external knowledge of a microcosm or a combination of those).
\RuleFont{Co-Sym} and \RuleFont{CR-Sym} are routine and legislate symmetry for $\parallel$ and $<>$.

There are two possible ways a microcosm can evolve upon receipt of new information:
(Section~\ref{Sect:Online.Decision} gives more details about the intuition and the semantics of the two evolution mechanisms.)
For a microcosm $M$, when $M \not\InitSays e_1\ \Wildcard\ e_2$, write $(M + e_1\ r\ e_2)$ for the microcosm $M$ with the additional information $e_1\ r\ e_2$.
We assume that $e_1\ r\ e_2$ is known to be internal or external to the resulting microcosm.
Likewise, when $M \InitSays e_1 <> e_2$ (or $M \InitSays e_2 <> e_1$), define $M[e_1 < e_2]$ for the microcosm $M$ in which $e_1 < e_2$ replaces $e_1 <> e_2$ (or $e_2 <> e_1$).

\begin{wrapfigure}{r}{0.56\textwidth}
  \vspace{-2em}
  \hrule
  \vspace{-0.6em}
  \begin{center}
    \begin{subfigure}[b]{0.25\columnwidth}
        \includegraphics[width=\textwidth]{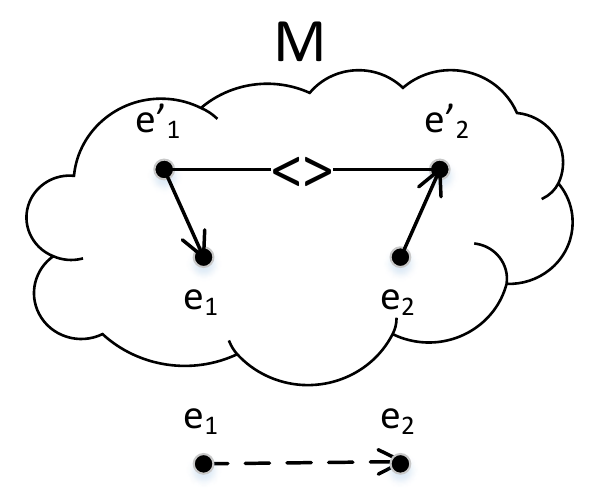}
        \caption{Addition of $e_1 < e_2$ to $M$ violates (\ref{Misc:M4}).}
        \label{Fig:NVC1}
    \end{subfigure}
    \begin{subfigure}[b]{0.25\columnwidth}
        \includegraphics[width=\textwidth]{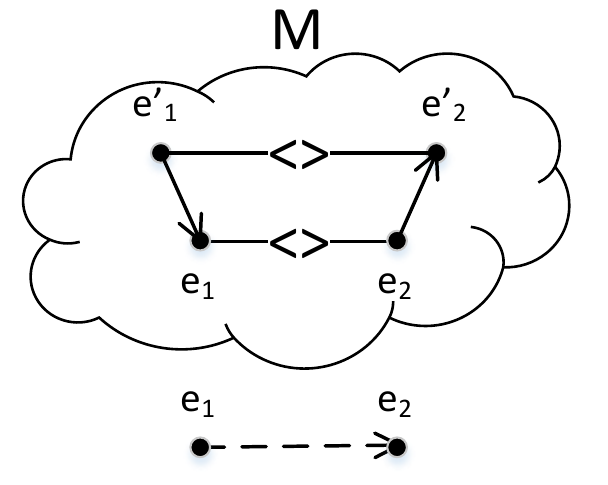}
        \caption{Updating $M$ with $e_1 < e_2$ violates (\ref{Misc:M4}).}
        \label{Fig:NVC2}
    \end{subfigure}
    \linebreak
    \underline{Legend}:
      Arrow from $e_1$ to $e_2$ indicates $e_1 < e_2$.
      Line labelled with ``$<>$'' between $e_1$ and $e_2$ indicates $e_1 <> e_2$.
  \end{center}
  \hrule
  \caption{Illegal for $e'_1 < e_1 < e_2 < e'_2$ whilst $e'_1 <> e'_2 \in M$}
  \label{Fig:NVCs}
  \vspace{-2em}
\end{wrapfigure}

For both addition -- namely, $(M + e_1\ r\ e_2)$ -- and update -- namely, $M[e_1 < e_2]$ -- we assume that the change will not violate (\ref{Misc:M4}).
See Fig.~\ref{Fig:NVC1} and \ref{Fig:NVC2} for when careless addition and update violate (\ref{Misc:M4}).
That can be considered a limitation in our model:
Once a microcosm reaches either state, further evolution via the given correspondence is banned by (\ref{Misc:M4}) forever.
We believe a batch addition (and update) can circumvent that limitation; of course, subject to sanity checks.
Take Fig.~\ref{Fig:NVC2} for example:
A batch update $M[e_1 < e_2, e'_1 < e'_2]$ is the simultaneous evolution of $M$ with both $e_1 < e_2$ and $e'_1 < e'_2$, in which (\ref{Misc:M4}) is also maintained when $e'_1 < e'_2$ is removed afterwards.
In this paper, we disregard batch addition and update and leave them to future.

The above discussion gives us context for explaining our design choice on including (\ref{Misc:M4}) in Definition~\ref{Defn:MCosm}.
Note that, without (\ref{Misc:M4}) outlawing it, after the addition of $e_1 < e_2$ to $M$ in Fig.~\ref{Fig:NVC1}, for instance, the initial judgements of the resulting microcosm would have become inconsistent:
On the one hand, $e'_1 <> e'_2$ would have been given; on the other hand, $e'_1 < e'_2$ would have been deducible by transitivity.

\section{Online Decision Making} \label{Sect:Online.Decision}

This section provides an algorithm for a microcosm to issue its verdict on the relation it can deduce, to the best of its knowledge, to hold between a queried pair of (distinct) events.
This algorithm (manifested in Fig.~\ref{Fig:Rules}) is called the \textbf{online} decision making procedure.
The idea is that the decision accuracy keeps improving using this procedure upon the inflow of the new or updated correspondences.
In crude terms, this is the situation where the device is connected and thus online.
Contrast this with what comes in Section~\ref{Sect:Offline.Decision}.
We prove computability (Theorem~\ref{Theo:Rules.Comp}) and consistency (Theorem~\ref{Theo:Micro.Rules.Const}) of the online decision making.
We show that the causal knowledge of a microcosm is strictly less than its enclosing world of events (Lemma~\ref{Lemm:WoE.More.Accurate.Micro}), there is no conflict between the verdict of two microcosms of the same world of events -- even when they do not issue the exact same correspondence (Corollary~\ref{Corr:Mu.Hats}).

\begin{figure*}
  \hrule
    $$
      \framebox[1.1\width]{$M \vdash e_1\ r\ e_2$} \qquad \textnormal{ where } r \in R \cup \{<>, ?\}
    $$
    $$
      \begin{prooftree}
          M \InitSays e_1\ r\ e_2
        \justifies
          M \vdash e_1\ r\ e_2
        \using
          \RuleFont{In-OK}
      \end{prooftree} \qquad
      \begin{prooftree}
          M \vdash e_1\ ?\ e_2
        \justifies
          M \vdash e_2\ ?\ e_1
        \using
          \RuleFont{Un-Sym}
      \end{prooftree}
    $$
    $$
      \begin{prooftree}
          \begin{array}{c@{\quad}c}
            M \vdash e_1\ r\ e_2       & M \vdash e_2\ r\ e_3\\
            r \in \{\parallel, <>, ?\} & M \not\InitSays e_1\ \Wildcard\ e_3
          \end{array}
        \justifies
          M \vdash e_1\ ?\ e_3
        \using
          \RuleFont{Un-1}
      \end{prooftree} \qquad
      \begin{prooftree}
          \begin{array}{c@{\quad}c}
            M \vdash e_1\ r\ e_2 & M \vdash e_2\ ?\ e_3\\
            r \in R \cup \{<>\} & M \not\InitSays e_1\ \Wildcard\ e_3
          \end{array}
        \justifies
          M \vdash e_1\ ?\ e_3
        \using
          \RuleFont{Un-2}
      \end{prooftree}
    $$
    $$
      \begin{prooftree}
          M \not\InitSays e_1\ \Wildcard\ e_2 \quad
          \nexists e' \in M.\ [(M \vdash e_1\ \Wildcard\ e') \wedge (M \vdash e'\ \Wildcard\ e_2)]
        \justifies
          M \vdash e_1\ ?\ e_2
        \using
          \RuleFont{Un-3} 
      \end{prooftree} \qquad
      \begin{prooftree}
          e \notin M \quad e \neq e'
        \justifies
          M \vdash e\ ?\ e'
        \using
          \RuleFont{Un-4}
      \end{prooftree}
    $$
    $$
      \begin{prooftree}
          M \vdash e_1\ ?\ e_2 \quad M \vartriangleleft W \quad W \vDash^* e_1\ r\ e_2
        \justifies
          (M + e_1\ r\ e_2) \vdash e_1\ r\ e_2
        \using
          \RuleFont{Strng}
      \end{prooftree} \quad
      \begin{prooftree}
          M \vdash e_1\ r\ e_2 \quad r \neq\ ? \quad M \vartriangleleft W \quad W \vDash^* e'_1\ r'\ e'_2 
        \justifies
          (M + e'_1\ r'\ e'_2) \vdash e_1\ r\ e_2
        \using
          \RuleFont{Weak}
      \end{prooftree}
    $$
    $$
      \begin{prooftree}
          M \vartriangleleft W \quad M \vdash e_1 <> e_2 \quad W \vDash e_1 < e_2
        \justifies
          M[e_1 < e_2] \vdash e_1 < e_2
        \using
          \RuleFont{Up-S}
      \end{prooftree}
    $$
    $$
      \begin{prooftree}
          M \vartriangleleft W \quad W \vDash e_1 < e_2 \quad M \vdash e_1 <> e_2 \quad M \vdash e'_1\ r'\ e'_2 \quad r'\ \neq\ ?
        \justifies
          M[e_1 < e_2] \vdash e'_1\ r'\ e'_2
        \using
          \RuleFont{Up-W}
      \end{prooftree}
    $$
  \hrule
  \caption{Online Decision Making}
  \vspace{-3em}
  \label{Fig:Rules}
\end{figure*}

\begin{definition} \label{Defn:Micro.Decision}
  Define the online decision making procedure of a microcosm using the rules in Fig.~\ref{Fig:Rules}.$\Square$
\end{definition}

In Fig.~\ref{Fig:Rules}, a judgement $M \vdash e_1\ ?\ e_2$ stipulates the lack of knowledge ``in $M$'' about the correspondence between the pair of events $e_1$ and $e_2$.
As such, $?$ (read \textsf{is-unknown-to}) is another inaccurate relation.
Note that, unlike $<>$, the relation $?$ is only available for microcosms.
Recall that, as axiomatised by (\ref{Misc:W5}), the correspondence between every two distinct pair of events is known to their enclosing world of events.

Here is an informal account of the rules in Fig.~\ref{Fig:Rules}:

\RuleFont{In-OK} says online decision making approves of initial judgements.
\RuleFont{Un-Sym} legislates symmetry of $?$.
The next two rows concern when a microcosm judges two events as unknown to one another.
\RuleFont{Un-1} decrees so when there is an intermediate event $e_2$ that has the same correspondence $r$ with both $e_1$ and $e_3$ (in different orders albeit).
Of course, given the transitivity of $<$, in the case of \RuleFont{Un-1}, $r$ cannot be $<$.
\RuleFont{Un-2} is similar except that, in the microcosm of discourse, the intermediate event $e_2$ is unknown to $e_3$.
Then, \RuleFont{Un-3} decrees for $e_1$ and $e_2$ to be unknown to one another when there is no intermediate event in the microcosm that is in correspondence with both $e_1$ and $e_2$.
The last rule of the group, i.e., \RuleFont{Un-4} declares the correspondence between an event that is not in a microcosm to be unknown with any other event.
Note that all the \RuleFont{Un-*} rules except \RuleFont{Un-4} assume that the microcosm has no initial judgements between the two events.

Next, the rules in the fourth row concern when a microcosm is supplied with new event information.
With such a supply, the microcosm of discourse evolves into a new one.
To this latter microcosm, one (and only one) more initial correspondence is available than the old microcosm.
Evolution happens either by strengthening or weakening.
\RuleFont{Strng} states that, when two events are judged to be unknown to one another by a microcosm, the judgement will be changed accordingly when the respective information from the enclosing world of event evolves the microcosm. 
\RuleFont{Weak} says the supply of new event information from the enclosing world of events preserves every event correspondence decreed earlier not to be unknown.
Note that the supply of new information is only possible via the enclosing world of event.

Finally, the last two rules are on update of $<>$ instances.
\RuleFont{Up-S} (for strengthening) and \RuleFont{Up-W} (for weakening) are the update counterparts \RuleFont{Strng} and \RuleFont{Weak}.
The difference is that, for the former pair of rules, the total number of correspondences initially known to the old microcosm and the new one are equal.
Yet, in \RuleFont{Up-S} and \RuleFont{Up-W}, one and only one $<>$ in the old microcosm is replaced by exactly one $<$ in the new microcosm.
The respective microcosm judgements are updated consequently.

The following lemma will later be used in Lemma~\ref{Lemm:Sample.Refutation.1}.

\begin{lemma} \label{Lemm:Online.Impl.Init}
  Suppose that $M \vdash e_1\ r\ e_2$, where $r \in R \cup \{<>\}$.
  Then, $M \InitSays e_1\ r\ e_2$.
\end{lemma}

Here are some notational conventions to be used shortly and thereafter:
We let $\Pi, \Pi', \dots, \Pi_1, \Pi_2, \dots$ range over derivation trees.
For a derivation tree $\Pi$, we write $\LastRuleOf{\Pi}$ for the last rule used in $\Pi$.
Additionally, for a derivation $\Pi$ of the form
$$
  \begin{prooftree}
      \Pi_1 \quad \Pi_2 \quad \dots \quad \Pi_n
    \justifies
      M\ \Wildcard\ c'
  \end{prooftree}
$$
\noindent we write $c \notin \Pi$ when $c \neq c'$ and $c \notin \Pi_1, c \notin \Pi_2, \ldots, c \notin \Pi_n$.
(The ``\Wildcard'' in ``$M\ \Wildcard\ c'$'' above can be ``$\InitSays$'', ``$\vdash$'', and ``$\vdash^*$.''
See Definition~\ref{Defn:Offline.Rules} for the latter.)
Intuitively, $c \not\in \Pi$ means that `$c$ does not appear in $\Pi$.'

\begin{lemma} \label{Lemm:No.Use.Shrink}
  Let $M \not\InitSays c$ and $\Pi = (M + c) \vdash e_1\ r\ e_2$ but $c \notin \Pi$.
  Then, $M \vdash e_1\ r\ e_2$.
\end{lemma}

Informally, the above lemma gives a criterion for shrinking the size of a derivation tree of online decisions.
Lemma~\ref{Lemm:No.Use.Shrink} will be used in the proof of Lemma~\ref{Lemm:Analog.No.Use}.

Fundamental results about online decision making follow.
Theorem~\ref{Theo:Rules.Comp} is on its computability.
Then, Theorem~\ref{Theo:Micro.Rules.Const} proves consistency.
At last, Lemma~\ref{Lemm:WoE.More.Accurate.Micro} and Corollary~\ref{Corr:Mu.Hats} focus on the relative accuracy of online decision making.

\begin{theorem} \label{Theo:Rules.Comp}
  Online decision making is computable:
  For every distinct pair of events $e_1$ and $e_2$ and microcosm $M$, in finite number of steps, the relation $r$ for which $M \vdash e_1\ r\ e_2$ can be found, if any.
\end{theorem}
\begin{proof}
  Let $p(e_1, e_2) = \exists r \in R \cup \{<>, ?\}.\ M \vdash e_1\ r\ e_2$.
  The result is trivial when no rule applies because, then, $p(e_1, e_2) = \bot$ in zero steps.
  Otherwise, we assume availability of a mechanism for preventing infinite trial of the symmetry rules.
  Similarly, we assume a book-keeping to prevent self-lookup over seeking an intermediate event in the case of \RuleFont{Un-$n$}, where $n \in \{1, 2, 3\}$.
  The proof is by rule-based induction on $M \vdash e_1\ r\ e_2$.
\end{proof}


Here is a note on the computational complexity of the online decision making.
Let $m$ be the size of a microcosm.
Taken na\"{i}vely, the rules in Fig.~\ref{Fig:Rules} give rise to exponential complexity w.r.t. $m$.
Using a simple $m \times m$ memoisation matrix, however, one can reduce that complexity to quadratic w.r.t. $m$.
Note that, having only a partial knowledge, being quadratic w.r.t. the size of a microcosm is far less than quadratic w.r.t. the size of a world of events as a whole.
This proves our model practically more useful than the classical holistic models.
In the presence of the above matrix, furthermore, maintaining (\ref{Misc:M4}) upon arrival of new correspondences is DLOG-Complete w.r.t. $m$ \cite[\S5.7]{Yap:1998}.

\begin{theorem} \label{Theo:Micro.Rules.Const}
  Online decision making is consistent: $M \vdash e_1\ r\ e_2$ and $M \vdash e_1\ r'\ e_2$ imply $r = r'$.
\end{theorem}
\begin{proof}
  Let $\Pi = M \vdash e_1\ r\ e_2$ and $\Pi' = M \vdash e_1\ r'\ e_2$.
  The proof is by rule-based induction on $\Pi$, namely, by case analysis of $\LastRuleOf{\Pi}$:
  \begin{itemize}

    \item \RuleFont{Un-$n$} for $n \in \{1, 2, 3\}$.\quad
      In all those cases, as a part of the hypotheses, $M \not\InitSays e_1\ \_\ e_2$.
      Hence, $\LastRuleOf{\Pi'} \neq \RuleFont{In-OK}$.
      Furthermore, $\LastRuleOf{\Pi'} \neq \RuleFont{Up-S}$ (because, then, $r' = <$ and $M \InitSays e_1 < e_2$) and $\LastRuleOf{\Pi'} \neq \RuleFont{Up-W}$ (because, then, $r' \neq\ ?$, and, by Lemma~\ref{Lemm:Online.Impl.Init}, $M \InitSays e_1\ r'\ e_2$).
      Likewise, $\LastRuleOf{\Pi'} \neq \RuleFont{Strng}$ either because, then, $M \InitSays e_1\ r'\ e_2$ for $M = (\Wildcard\ + e_1\ r'\ e_2)$.
      We claim that the last rule in $M \vdash e_1\ r'\ e_2$ cannot be \RuleFont{Weak} either, and, the result follows because all the remaining rules imply that $r'\ =\ ?$.

      We now prove our last claim.
      If the last rule in $M \vdash e_1\ r'\ e_2$ is to be \RuleFont{Weak}, there exists a microcosm $M'$ such that $M = (M' + e'_1\ \Wildcard\ e'_2)$ and $M' \vdash e_1\ r'\ e_2$.
      Besides, $r'\ \neq\ ?$, which, by Lemma~\ref{Lemm:Online.Impl.Init}, implies $M' \InitSays e_1\ r'\ e_2$.
      This is, however, a contradiction because, then $M \InitSays e_1\ r'\ e_2$.

    \item \RuleFont{Un-4}.\quad
      When $e \notin M$ and $e \neq e'$, there essentially is no other rule that can apply than \RuleFont{Un-4}.
      That is, the last rule for $M \vdash e_1\ r'\ e_2$ too needs to be \RuleFont{Un-4} and $r'\ =\ ?$.

%

  \end{itemize}
  We drop the remaining cases due to space restrictions.
\end{proof}

\begin{definition} \label{Defn:More.Accurate}
  For a pair of relations $r, r' \in R \cup \{<>, ?\}$, write $r \sqsubset r'$ -- for $r$ is at most as accurate as $r'$ -- when:
  (i) $r'\ \neq\ ?$ and $r\ =\ ?$,
  (ii) $r'\ =\ <$ and $r\ = <>$,
  (iii) $r'\ =\ <^{-1}$ and $r\ =\ <>$,
  (iv) $r\ = r'\ =\ <$, and
  (v) $r\ = r'\ =\ \parallel$.
  Write $\sqsubseteq$ for the reflexive closure of $\sqsubset$.
\end{definition}

The following result states that a microcosm always \textbf{approximate}s its enclosing world of event:
For every pair of events, when the relation a microcosm attributes to the pair does not exactly coincide with that of its enclosing world of events, the microcosm is only less accurate.
This is the essence of our model being weaker than the mainstream practice where every device is exactly as accurate as its enclosing ecosystem.

\begin{lemma} \label{Lemm:WoE.More.Accurate.Micro}
  Let $M \vartriangleleft W$.
  Suppose also that $W \vDash e_1\ r_W\ e_2$ and $M \vdash e_1\ r_M\ e_2$.
  Then, $r_M \sqsubset r_W$.
\end{lemma}

The next result says:
When two microcosms of the same world of events do not agree on a given pair of events, it only is that one of the two is more accurate than the other.
In other words, two microcosms of the same world of events will never attribute conflicting relations to any given pair of events.

\begin{corollary} \label{Corr:Mu.Hats}
  Let $M \vartriangleleft W$ and $M' \vartriangleleft W$ with $M \vdash e_1\ r\ e_2$ and $M' \vdash e_1\ r'\ e_2$.
  Then, $r \sqsubseteq r'$ or $r' \sqsubseteq r$.
\end{corollary}

\section{Offline Decision Making} \label{Sect:Offline.Decision}

The algorithm presented in this section enables a microcosm to make new decisions without depending on new supply from the enclosing world of events.
As such, it suits a device required to perform offline computation.
Hence, the naming ``offline.''
Unlike our online algorithm that exclusively proves correspondences, our offline algorithm is based on cancelling possibilities.
That is, deducing it that certain correspondences cannot possibly hold between the given pair of events.
We say that the online decision making \textit{confirm}s, whereas the offline one (mostly) \textit{refute}s.

Sometimes, cancelling enough possibilities out will prove the only remaining correspondence (e.g., Fig.~\ref{Fig:WO2}).
But, even if that is not quite the case, cancelling one or more correspondences out is still useful (e.g., Fig.~\ref{Fig:WO1}):
It conveys the information that the given pair of events are \textbf{not} unknown to one another.
(See Lemma~\ref{Lemm:Sample.Refutation.1}.)
Most particularly, in such a scenario, it would be wrong to consider the pair concurrent.
That is in exact contrast with the common causality folklore that: `when one cannot confirm any correspondence between two events, one can safely [sic] consider them concurrent.'

In Fig.~\ref{Fig:WO1}, given that $e_1 \parallel e_2$ and $e_2 < e_3$, it cannot be that $e_3 < e_1$.
This is because, then, by transitivity of \textsf{happens-before}, $e_2 < e_3$ and $e_3 < e_1$, imply $e_2 < e_1$, contradicting $e_1 \parallel e_2$.
Fig.~\ref{Fig:WO2} rules $e_3 < e_1$ out similarly.
But, then, given that $e_1 <> e_3$, the implication is $e_1 < e_3$.
Note that the only correspondences that were available prior to concluding $e_3 \not< e_1$ (in Fig.~\ref{Fig:WO1}) and $e_1 < e_3$ (in Fig.~\ref{Fig:WO2}) were the black lines between $e_1$, $e_2$, and $e_3$.
No new correspondence was supplied over the arguments either.
The important observation to make, hence, is that such arguments do not depend on new supply from the enclosing world of events.
Offline decision making (Definition~\ref{Defn:Offline.Rules}) enables such arguments.

\begin{figure*}
  \hrule
  \centering
  \begin{subfigure}[b]{0.4\textwidth}
    \underline{Legend}:
    Two parallel lines between $e$ and $e'$ depicts $e \parallel e'$.
    Dotted arrows show hypothesised \textsf{happens-before}.
    Red lines show what goes wrong as a result of the hypotheses.
    Green arrow shows \textsf{happens-before} that was proved offline.
  \end{subfigure}
  \hfill
  \vline
  \hfill
  \begin{subfigure}[b]{0.27\textwidth}
      \centering
      \includegraphics[width=0.59\textwidth]{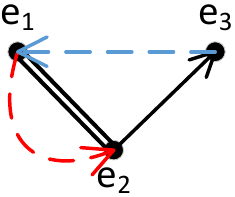}
      \caption{$[(e_1 \parallel e_2) \wedge (e_2 < e_3)] \Rightarrow e_3 \not< e_1$}
      \label{Fig:WO1}
  \end{subfigure}
  \hfill
  \vline
  \hfill
  \begin{subfigure}[b]{0.29\textwidth}
      \centering
      \includegraphics[width=0.55\textwidth]{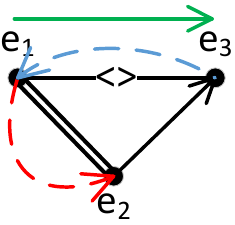}
      \caption{$[(e_1 \parallel e_2) \wedge (e_2 < e_3) \wedge (e_3 <> e_1)] \Rightarrow e_1 < e_3$}
      \label{Fig:WO2}
  \end{subfigure}
  \hrule
  \vspace{0.1em}
  \caption{Two Useful Offline Deductions}
  \label{Fig:Why.Offline}
  \vspace{-2em}
\end{figure*}

Before we can delve into offline decision making itself, we need to introduce a couple of notations.
For a microcosm $M$ and a pair of distinct events $e_1$ and $e_2$ such that $M \not\vdash e_1\ \Wildcard\ e_2$, write $(M \Supposing e_1\ r\ e_2)$ for a microcosm that is \textit{structurally} the same as $(M + e_1\ r\ e_2)$, namely, contains the exact same correspondences.
Despite their same structure, the former is meant to be used only when $e_1\ r\ e_2$ is not supplied by the enclosing world of events; it rather is $M$ with the \textit{hypothesis} that $e_1\ r\ e_2$ was also known by $M$.
That is, ``$\Supposing$'' is like the blue arrow in Fig.~\ref{Fig:WO1}.
Note the additional requirement of the former over the latter.
The latter only requires that $M \not\InitSays e_1\ \Wildcard\ e_2$.
In contrast, the former requires that $M \not\vdash e_1\ \Wildcard\ e_2$.
(By definition, the requirement for $(M \Supposing e_1\ r\ e_2)$ implies the requirement of $(M + e_1\ r\ e_2)$ too.
Hence, $(M \Supposing e_1\ r\ e_2)$ is well-defined.)
Note also that, by Theorem~\ref{Theo:Rules.Comp}, satisfiability of $M \not\vdash e_1\ \Wildcard\ e_2$ is computable.
When $M \not\vdash e_1\ \Wildcard\ e_2$, define $M\SuppUp{e_1\ r\ e_2}$ similarly for a microcosm that is structurally like $M[e_1\ r\ e_2]$; yet $e_1\ r\ e_2$ is not supplied by the enclosing world of events but is only hypothesised.
That is, ``$\SuppUp{.}$'' is like the blue arrow in Fig.~\ref{Fig:WO2}.

\begin{figure*}
  \hrule
  \vspace{0.1em}
  \begin{center}
    \framebox[1.1\width]{$M \vdash^* e_1\ \tilde{r}\ e_2$} where $\tilde{r} ::= r \mid \NotRel$
  \end{center}
  $$
    \begin{prooftree}
        M \vdash e_1\ r\ e_2
      \justifies
        M \vdash^* e_1\ r\ e_2
      \using
        \RuleFont{Onl-OK}
    \end{prooftree}\qquad
    \begin{prooftree}
        M \vdash^* e_1 \NotRel\ e_2 \quad r \in R
      \justifies
        M \vdash^* e_1 \NotUnknownTo\ e_2
      \using
        \RuleFont{Not-R}
    \end{prooftree}
  $$
  $$
    \begin{prooftree}
        M \vdash^* e_1\ \cancel{<>}\ e_2
      \justifies
        M \vdash^* e_1 \parallel e_2
      \using
        \RuleFont{Not-CR}
    \end{prooftree}\qquad
    \begin{prooftree}
        M \vdash^* e_1\ \NotConc\ e_2
      \justifies
        M \vdash^* e_1\ <>\ e_2
      \using
        \RuleFont{Not-Co}
    \end{prooftree}
  $$
  $$
    \begin{prooftree}
        (M \Supposing e_1\ r\ e_2) \vdash^* e'_1\ r'_1\ e'_2 \quad (M \Supposing e_1\ r\ e_2) \vdash^* e'_1\ r'_2\ e'_2 \quad r_1 \neq r_2
      \justifies
        M \vdash^* e_1 \NotRel\ e_2
      \using
        \RuleFont{Cntrd}
    \end{prooftree}
  $$
  $$
    \begin{prooftree}
        M\SuppUp{e_1\ r\ e_2} \vdash^* e'_1\ r'_1\ e'_2 \quad M\SuppUp{e_1\ r\ e_2} \vdash^* e'_1\ r'_2\ e'_2 \quad r_1 \neq r_2
      \justifies
        M \vdash^* e_1 \NotRel\ e_2
      \using
        \RuleFont{Up-Cntrd}
    \end{prooftree}
  $$
  $$
    \begin{prooftree}
        M \vdash^* e_1\ <>\ e_2 \quad M \vdash^* e_1\ \not < e_2
      \justifies
        M \vdash^* e_2 < e_1
      \using
        \RuleFont{Not-HB}
    \end{prooftree}\qquad
    \begin{prooftree}
        M \vdash^* e_1 \not < e_2 \quad M \vdash^* e_2 \not < e_1
      \justifies
        M \vdash^* e_1\ \cancel{<>}\ e_2
      \using
        \RuleFont{No-HBs}
    \end{prooftree}
  $$
  $$
    \begin{prooftree}
        M \vdash^* e_1 \NotUnknownTo\ e_2
      \justifies
        M \vdash^* e_2 \NotUnknownTo\ e_1
      \using
        \RuleFont{NU-Sym}
    \end{prooftree}\qquad
    \begin{prooftree}
        M \vdash^* e_1\ \cancel{<>}\ e_2
      \justifies
        M \vdash^* e_2\ \cancel{<>}\ e_1
      \using
        \RuleFont{NCR-Sym}
    \end{prooftree}\qquad
    \begin{prooftree}
        M \vdash^* e_1 \NotConc e_2
      \justifies
        M \vdash^* e_2 \NotConc e_1
      \using
        \RuleFont{NCo-Sym}
    \end{prooftree}
  $$
  \hrule
  \caption{Offline Decision Making}
  \vspace{-3em}
  \label{Fig:Offline.Rules}
\end{figure*}

\begin{definition} \label{Defn:Offline.Rules}
  Define the offline decision making procedure of a microcosm using the rules in Fig.~\ref{Fig:Offline.Rules}, where the judgements take the form $M \vdash^* e_1\ \tilde{r}\ e_2$, and $\tilde{r} ::= r\ \mid \not{r}$.
\end{definition}

The rules in Fig.~\ref{Fig:Offline.Rules} are fairly self-explanatory and we drop explanation to save space, except for the two key rules: \RuleFont{Cntrd} and \RuleFont{Up-Cntrd}.
If a hypothetical correspondence between a pair of events leads to two different conclusions about a single pair of distinct events, we have come to a contradiction, and, conclude the hypothesis to be false.
\RuleFont{Cntrd} manifests that for additions whilst \RuleFont{Up-Cntrd} does so for updates.

The online and offline decisions on the same pair of events will not conflict.
That is, online and offline decision making are consistent:

\begin{lemma} \label{Lemm:Offline.Not.Refute.Online}
  Let $e_1$ and $e_2$ be a pair of distinct events in $M$.
  Then: (i) $M \vdash e_1\ r\ e_2$ implies $M \not\vdash^* e_1 \NotRel\ \ e_2$, and
  (ii) $M \vdash^* e_1 \NotRel\ e_2$ implies $M \not\vdash e_1\ \Wildcard\ e_2$, in particular, $M \not\vdash e_1\ r\ e_2$.
\end{lemma}

\begin{lemma} \label{Lemm:Offline.Not.Disagree.Online}
  If $M \vdash e_1\ r\ e_2$ and $M \vdash^* e_1\ r'\ e_2$, then $r = r'$.
\end{lemma}

The offline decision making can be used, for example, to mechanically conclude in the case of Fig.~\ref{Fig:WO1} that $M \vdash^* e_3 \NotUnknownTo\ e_1$:

\begin{lemma} \label{Lemm:Sample.Refutation.1}
  Let $M \vdash e_1 \parallel e_2$ and $M \vdash e_2 < e_3$ but $M \not\vdash e_3\ \Wildcard\ e_1$.
  Then, $M \vdash^* e_3 \NotUnknownTo\ e_1$.
\end{lemma}
\begin{proof}
  The mechanical proof comes in Fig.~\ref{Fig:Sample.Refutation}, where the derivation labelled ($*$) is Lemma~\ref{Lemm:Online.Impl.Init}.
\end{proof}

\begin{figure*}
  \hrule
  \vspace{0.1em}
  \begin{center}
    \includegraphics[width=\textwidth]{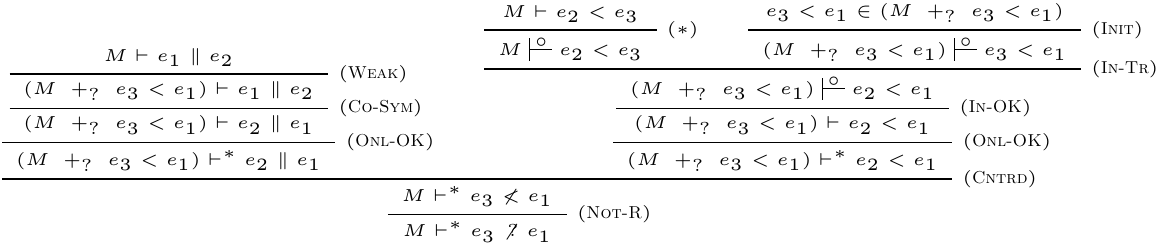}
  \end{center}
  \vspace{-0.5em}
  \hrule
  \caption{Mechanical Proof of Lemma~\ref{Lemm:Sample.Refutation.1}}
  \vspace{-2.5em}
  \label{Fig:Sample.Refutation}
\end{figure*}

Despite its merits, offline decision making is confronted with two problems:
Firstly, getting to refute the right correspondence may only be possible using human intelligence.
Although mechanical proofs like Fig.~\ref{Fig:Sample.Refutation} help a human-being legislate informal reasoning such as Fig.~\ref{Fig:Why.Offline}, how likely that is for a machine to produce that is not yet known.
Secondly, the search space for getting to a contradiction (and hence a refute) is exponential in the number of events known to a microcosm.
We are not aware of any technique for reducing that space.

\section{Forward Bisimilarity} \label{Sect:ForBiSim}

In this section, we present our first notion of microcosm bisimilarity.
We start by defining microcosm analogy (Definition~\ref{Defn:Micro.Analog}), namely, what exactly we mean when we say two microcosms agree on every correspondence.
Then, we show that such microcosms will evolve likewise when supplied with the exact same new single correspondence (Theorem~\ref{Theo:Analog.Bisimu}), i.e., they are forward bisimilar (Definition~\ref{Defn:For.Bisimilarity}).
The most important result of this section is Theorem~\ref{Theo:For.Order.Irrelevant}, which proves it that the order of arrival of causal information is irrelevant so long as the same correspondences are available to a pair of bisimilar microcosms.
Finally, Theorem~\ref{Theo:For.Bisimilarity} establishes the bisimilarity of analogy.
We call the bisimilarity of this section forward to contrast it with that of next section (Definition~\ref{Defn:Back.Bisimilarity}), which we call backward.

\begin{definition} \label{Defn:Micro.Analog}
  Call microcosms $M$ and $M'$ analogous -- write $M \approx M'$ -- when: $\forall e_1, e_2.\ M \vdash e_1\ r\ e_2 \Leftrightarrow M' \vdash e_1\ r\ e_2$.
\end{definition}

In words, two microcosms are analogous when they `agree on the correspondence between every pair of events.'
That can, for instance, be two replicas of a single data centre that are in the same state.
As another example consider a copy taken from a device before it temporarily dies.
As soon as the original device comes back to life, the original device and the copy would be analogous.
Interestingly enough, the order of arrival of the causal information to the original device is completely sporadic to the copy.
Note that Definition~\ref{Defn:Micro.Analog} has even no explicit mention of the enclosing worlds of events of the two microcosms.

\begin{definition} \label{Defn:For.Trans.Sys}
  Define \SingleC\ for the transition system $\ForwardTS = (\mathcal{M}_W, \mathcal{C}^*_W, \stackrel{.}{\rightarrow})$ such that $M \SingleC M'$ when $M' = (M + c)$ for some $c \in \mathcal{C}^*_W$.
  Call \ForwardTS\ the forward transition system of $W$.
\end{definition}

The above definition formalises our understanding of a microcosm evolving \textit{forward} with the arrival of new supply to it.
We now present a technical lemma for later use.

\begin{lemma} \label{Lemm:Analog.Init.Use}
  Suppose that $M \approx M'$.
  Then, $\Pi = (M + c) \InitSays e_1\ r\ e_2$ implies $(M' + c) \InitSays e_1\ r\ e_2$ when $c \in \Pi$.
\end{lemma}
\begin{proof}
  By case distinction on $\LastRuleOf{\Pi}$.
\end{proof}

The following two lemmata explore two different scenarios for forward evolution:
when the new supply is not used for deriving the correspondence between a given pair of events (Lemma~\ref{Lemm:Analog.No.Use}) and when it is (Lemma~\ref{Lemm:Analog.Use}).
Those two pave the road for Theorem~\ref{Theo:Analog.Bisimu}.

\begin{lemma} \label{Lemm:Analog.No.Use}
  Suppose that $M \approx M'$ and $M \SingleC (M + c)$.
  Then, $\Pi = (M + c) \vdash e_1\ r\ e_2$ implies $(M' + c) \vdash e_1\ r\ e_2$ when $c \notin \Pi$.
\end{lemma}
\begin{proof}
  By Lemma~\ref{Lemm:No.Use.Shrink}, $(M + c) \vdash e_1\ r\ e_2$ and $c \notin \Pi$ imply $M \vdash e_1\ r\ e_2$.
  Given that $M \approx M'$, thus, $M' \vdash e_1\ r\ e_2$.
  Note, then, that $M' \not\vdash e_1\ \Wildcard\ e_2$ is not true.
  We claim that the result follows, i.e., $(M' + c) \vdash e_1\ r\ e_2$.

  Here is the proof of our claim.
  Suppose otherwise, namely, that $(M' + c) \vdash e_1\ r'\ e_2$ where $r \neq r'$.
  Then, the only rule that can take $M' \vdash e_1\ r\ e_2$ to $(M' + c) \vdash e_1\ r'\ e_2$ is \RuleFont{Strng}, in which case, $r\ =\ ?$ and $r' \in R \cup \{<>\}$.
  But, then, $M \vdash e_1\ ?\ e_2$ for $M \approx M'$.
  It follows by \RuleFont{Strng} that $(M + c) \vdash e_1\ r'\ e_2$.
  That, however, is a contradiction because, according to Theorem~\ref{Theo:Micro.Rules.Const}, $\vdash$ is consistent.
\end{proof}

\begin{lemma} \label{Lemm:Analog.Use}
  Suppose that $M \approx M'$ and $M \SingleC (M + c)$.
  Then, $\Pi = (M + c) \vdash e_1\ r\ e_2$ implies $(M' + c) \vdash e_1\ r\ e_2$ when $c \in \Pi$.
\end{lemma}
\begin{proof}
  Induction on the size of $\Pi$ by case distinction on \LastRuleOf{\Pi}.
  \begin{itemize}
    \item \RuleFont{In-OK}.\quad
      In this case, $(M + c) \InitSays e_1\ r\ e_2$.
      By Lemma~\ref{Lemm:Analog.Init.Use}, $(M' + c) \InitSays e_1\ r\ e_2$.
      Using an application of \RuleFont{In-OK}, one derives the desirable.

    \item \RuleFont{Un-1}.\quad
      In this case, $r\ =\ ?$, and, there exist $r' \in R \cup \{<>\}$ and an intermediate event $e$ such that $\Pi_1 = (M + c) \vdash e_1\ r'\ e$ and $\Pi_2 = (M + c) \vdash e\ r'\ e_2$.
      When $c \notin \Pi_1$, by Lemma~\ref{Lemm:Analog.No.Use}, $(M' + c) \vdash e_1\ r'\ e$.
      Otherwise, the same result is obtained by the inductive hypothesis.
      Based on whether $c \in \Pi_2$ or not, one obtains $(M' + c) \vdash e\ r'\ e_2$ similarly.
      The other hypothesis of this rule is $(M + c) \not\InitSays e_1\ \Wildcard\ e_2$.
      We claim that $(M' + c) \not\InitSays e_1\ \Wildcard\ e_2$.
      One derives the desirable using an application of \RuleFont{Un-1}.

      Here is the proof of our claim.
      Suppose otherwise.
      Then, $\Pi_\circ = (M' + c) \InitSays e_1\ r_\circ\ e_2$ for some $r_\circ \in R \cup \{<>\}$.
      When $c \in \Pi_\circ$, by Lemma~\ref{Lemm:Analog.Init.Use}, $(M + c) \InitSays e_1\ r_\circ\ e_2$.
      When $c \notin \Pi_\circ$, using an application of \RuleFont{In-OK}, one first gets $(M' + c) \vdash e_1\ r_\circ\ e_2$.
      Next, one uses Lemma~\ref{Lemm:Analog.No.Use} to conclude $(M + c) \vdash e_1\ r_\circ\ e_2$.
      Both cases, however, contradict consistency of $\vdash$ (c.f., Theorem~\ref{Theo:Micro.Rules.Const}).

    \item \RuleFont{Un-4}.\quad
      This case is not applicable.
      Here is why.
      Suppose otherwise.
      Then, $r\ =\ ?$.
      Furthermore, the only way for $c \in \Pi$ is that $c = e_1\ ?\ e_2$.
      But, that is not possible because, by Definition~\ref{Defn:For.Trans.Sys}, $M \SingleC (M + c)$ is only defined when $c \in \mathcal{C}^*_W$.
      That is, $W \vDash^* e_1\ ?\ e_2$, which cannot be.

    \item \RuleFont{Strng}.\quad
      In this case, $M \vdash e_1\ ?\ e_2$.
      Given that $M \approx M'$, it follows that $M' \vdash e_1\ ?\ e_2$.
      Using an application of \RuleFont{Strng}, then, $(M' + c) \vdash e_1\ r\ e_2$.

    \item \RuleFont{Weak}.\quad
      In this case, $M \vdash e_1\ r\ e_2$.
      Given that $M \approx M'$, it follows that $M' \vdash e_1\ r\ e_2$.
      Using an application of \RuleFont{Weak}, then, $(M' + c) \vdash e_1\ r\ e_2$.
  \end{itemize}
  We drop the remaining cases due to space restrictions.
\end{proof}

We next define our notion of forward bisimulation and prove that analogy is a bisimulation.

\begin{definition} \label{Defn:For.Bisimilarity}
  Call a binary relation $\mathpzc{R}$ on $\mathcal{M}_W$ a bisimulation for $\ForwardTS$ when for every microcosms $M_1$ and $M_2$ of $W$ such that $M_1\ \mathpzc{R}\ M_2$, the following hold:
  \begin{itemize}
    \item $M_1 \SingleC M'_1 \Rightarrow \exists M'_2 \vartriangleleft W.\ (M_2 \SingleC M'_2) \wedge (M'_1\ \mathpzc{R}\ M'_2)$, and
    \item $M_2 \SingleC M'_2 \Rightarrow \exists M'_1 \vartriangleleft W.\ (M_1 \SingleC M'_1) \wedge (M'_1\ \mathpzc{R}\ M'_2)$.
  \end{itemize}
  Write \ForBiSim\ for the bisimilarity of $\ForwardTS$, i.e., the largest bisimulation for $\ForwardTS$.
\end{definition}

\begin{theorem} \label{Theo:Analog.Bisimu}
  For every $W$, the relation $\approx$ is a bisimulation for $\ForwardTS$.
\end{theorem}
\begin{proof}
  Let $M, M' \vartriangleleft W$ and $M \approx M'$.
  Suppose that $M \SingleC (M + c)$ and $\Pi = (M + c) \vdash e_1\ r\ e_2$.
  When $c \notin \Pi$, by Lemma~\ref{Lemm:Analog.No.Use}, $(M + c) \vdash e_1\ r\ e_2$.
  When $c \in \Pi$, by Lemma~\ref{Lemm:Analog.Use}, $(M + c) \vdash e_1\ r\ e_2$.
  The result follows by symmetry.
\end{proof}

Now that we are armed with Theorem~\ref{Theo:Analog.Bisimu}, it is easy to prove Theorem~\ref{Theo:For.Order.Irrelevant}.
We would like to draw the reader's attention to the small length of the proof and the simple technique used for it.
Such a comfort is a consequence of bisimulation being such a strong concept.

For a given $n$, write $\overline{c}$ for $c_1, c_2, \dots, c_n$ and $n = |\overline{c}|$.
Extend $\stackrel{.}{\rightarrow}$, accordingly, to $\stackrel{.}{\twoheadrightarrow}$ where \MultiCs\ abbreviates $\SingleCNamed{c_1} \circ \SingleCNamed{c_2} \circ \dots \circ \SingleCNamed{c_n}$.
Furthermore, write $\overline{c}' = p(\overline{c})$ when $\overline{c}'$ is a permutation of $\overline{c}$.

\begin{theorem} \label{Theo:For.Order.Irrelevant}
  Suppose that $M_0 \approx M'_0$.
  Suppose also that $M_0 \MultiCs M$ and $M'_0 \MultiCsPed M'$, where $\overline{c}' = p(\overline{c})$.
  Then, $M \approx M'$.
\end{theorem}
\begin{proof}
  We proceed by strong induction on $n$, where $n = |\overline{c}|$:
  \begin{itemize}
    \item $n = 1$.\quad
      By Theorem~\ref{Theo:Analog.Bisimu}.
    
    \item $n = k$.\quad
      Suppose that the theorem is correct for every $n < k$.
      The case when $\overline{c} = \overline{c}'$ is immediate.
      Otherwise, let $k_0$ be the first position where $\overline{c}$ and $\overline{c}'$ disagree.
      That is, $M_0 \MultiCsNamed{\overline{c}_l} M_{k_0 - 1} \SingleCNamed{c_{k_0}} M_{k_0} \MultiCsNamed{\overline{c}_r} M$ and $M'_0 \MultiCsNamed{\overline{c}'_l} M'_{k_0 - 1} \SingleCNamed{c'_{k_0}} M'_{k_0} \MultiCsNamed{\overline{c}'_r} M'$ such that $\overline{c}_l = \overline{c}'_l$, $c_{k_0} \neq c'_{k_0}$, and $\overline{c}'_r = p(\overline{c}_r)$.
      Then, $M_{k_0} \approx M'_{k_0}$ is immediate from Theorem~\ref{Theo:Analog.Bisimu}.
      And, given that $|c_{k_0}\overline{c}_r| = |c'_{k_0}\overline{c}'_r| < k$, by the inductive hypothesis, $M \approx M'$.
  \end{itemize}
  The result follows.
\end{proof}

\begin{theorem} \label{Theo:For.Bisimilarity}
  For $W$, the relation $\approx$ is the bisimilarity of $\ForwardTS$, i.e., $\approx = \sim_F$.
\end{theorem}
\begin{proof}
  Given that $\ForBiSim$ is the largest bisimulation of $\ForwardTS$ (c.f., Definition~\ref{Defn:For.Bisimilarity}), it suffices that we show $\ForBiSim\ \subseteq\ \approx$.
  To that end, suppose that $M \ForBiSim M'$; we will show that $M \approx M'$.
  Choose an arbitrary pair of events $e_1$ and $e_2$ such that $\Pi = M \vdash e_1\ r\ e_2$ and $\Pi' = M \vdash e_1\ r'\ e_2$.
  The proof is by parallel induction on $\Pi$ and $\Pi'$ and proceeds by case distinction on $\LastRuleOf{\Pi}$ and $\LastRuleOf{\Pi'}$.
  The goal is to show that, in all the possible cases, $r = r'$.
  We only show one case here:
  \begin{itemize}
    \item $M \vdash e_1\ r\ e_2$ and $M' \vdash e_1\ r'\ e_2$ but $\{r, r'\} \cap \{?\} = \varnothing$.\quad
      Note first that, in this case, by Corollary~\ref{Corr:Mu.Hats}, regardless of $M$ and $M'$ being bisimilar, either $r \sqsubseteq r')$ or $r' \sqsubseteq r$.
      We now show that, in the case of bisimilarity, $r = r'$.
      Let us assume that $r \sqsubseteq r'$; let us also assume that $r\ =\ <$ and $r'\ =\ <>$; the proof is similar otherwise.

      When $M \vdash e_1 < e_2$ but $M' \vdash e_1\ <>\ e_2$, by Lemma~\ref{Lemm:Online.Impl.Init}, $M \InitSays e_1 < e_2$ and $M'\InitSays e_1\ <>\ e_2$, respectively.
      Hence, for an event $e_3$ such that $W \vDash e_2 < e_3$, it follows using an application of \RuleFont{In-Tr} that $M \InitSays e_1 < e_3$.
      Using an application of \RuleFont{In-OK}, then, $M \vdash e_1 < e_3$.
      This is whilst, with the given information, $M' \InitSays e_1\ \Wildcard\ e_3$ is not derivable.
      Thus, $(M' + e_1 < e_3)$ is defined, but, $(M' + e_1 < e_3)$ is not.
      Let $c = e_2 < e_3$.
      Then, $M' \SingleC (M' + c)$ but $M \SingleC \Wildcard$ is not implied, which contradicts $M \ForBiSim M'$.
      (See Definition~\ref{Defn:For.Bisimilarity}.)
  \end{itemize}
  We omit the remaining cases due to space restrictions.
\end{proof}

A shortcoming of $\ForBiSim$ is that it only studies microcosm evolution via addition.
Whereas microcosms can well evolve via update too.
We leave the study of $\ForBiSim$ in presence of updates (as well as additions) to future work.
The same applies to $\BackBiSim$, which we will consider next.

\section{Backward Bisimilarity}

Only limited resources are available to devices, especially the edge devices.
Emptying the disk or memory of such a device is routine then.
To that end, usually, one removes the outdated data to come to a new manageable state.
This section deals with when (causal) information is to be removed from devices, say due to resource limitation or outdatedness.
That too can be seen as an evolution for a microcosm, albeit \textit{backward} (Definition~\ref{Defn:Back.Trans.Sys}).
We show that microcosm analogy (Definition~\ref{Defn:Micro.Analog}) gives rise to a bisimilarity for backward evolution as well (Theorem~\ref{Theo:Back.Bisimilarity}).
Besides, this section presents the backward counterpart of Theorem~\ref{Theo:For.Order.Irrelevant} that proves:
The order of removal of causal information from bisimilar devices does not matter in that they will again be bisimilar once they are both done with the set of correspondences (Theorem~\ref{Theo:Back.Order.Irrelevant}).

\begin{definition} \label{Defn:Back.Trans.Sys}
  Define \SingleBC\ for the transition system $\BackwardTS = (\mathcal{M}_W, \mathcal{C}^*_W, \stackrel{.}{\leftarrow})$ such that $M \SingleBC M'$ when $M = (M' + c)$ for some $c \in \mathcal{C}^*_W$.
  Call \BackwardTS\ the backward transition system of $W$.
\end{definition}

The notation $M \SingleBC\ M'$ is indeed intended to be read from left to right to denote getting from $M$ to $M'$ by the removal of $c$.

\begin{definition} \label{Defn:Back.Bisimilarity}
  Call a binary relation $\mathpzc{R}$ on $\mathcal{M}_W$ a bisimulation for $\BackwardTS$ when for every microcosms $M_1$ and $M_2$ of $W$ such that $M_1\ \mathpzc{R}\ M_2$, the following hold:
  \begin{itemize}
    \item $M_1 \SingleBC M'_1 \Rightarrow \exists M'_2 \vartriangleleft W.\ (M_2 \SingleBC M'_2) \wedge (M'_1\ \mathpzc{R}\ M'_2)$, and
    \item $M_2 \SingleBC M'_2 \Rightarrow \exists M'_1 \vartriangleleft W.\ (M_1 \SingleBC M'_1) \wedge (M'_1\ \mathpzc{R}\ M'_2)$.
  \end{itemize}
  Write \BackBiSim\ for the bisimilarity of $\BackwardTS$, i.e., the largest bisimulation for $\BackwardTS$.
\end{definition}

\begin{theorem} \label{Theo:Back.Bisimu}
  For every $W$, the relation $\approx$ is a bisimulation for $\BackwardTS$.
\end{theorem}

We extend $\stackrel{.}{\leftarrow}$, like $\stackrel{.}{\twoheadrightarrow}$ to $\stackrel{.}{\twoheadleftarrow}$ where \MultiBCs\ abbreviates $\SingleBCNamed{c_1} \circ \SingleBCNamed{c_2} \circ \dots \circ \SingleBCNamed{c_n}$.
In words, the following theorem states that the order of removal is irrelevant so long as the same set of correspondences are removed from analogous microcosms.

\begin{theorem} \label{Theo:Back.Order.Irrelevant}
  Suppose that $M_0 \approx M'_0$.
  Suppose also that $M_0 \MultiBCs M$ and $M'_0 \MultiBCsPed M'$, where $\overline{c}' = p(\overline{c})$.
  Then, $M \approx M'$.
\end{theorem}
\begin{proof}
  Similar to Theorem~\ref{Theo:For.Order.Irrelevant}.
\end{proof}

\begin{theorem} \label{Theo:Back.Bisimilarity}
  For $W$, the relation $\approx$ is the bisimilarity of $\BackwardTS$, i.e., $\approx = \sim_B$.
\end{theorem}

\section{Related Work} \label{Sect:Related}

The partial knowledge of a microcosm w.r.t. its enclosing world of events resembles the classical ``knowledge vs common knowledge'' model \cite{Halp+Mose:1990,Fagi+Halp+Mose+Vard:2003}.
The latter works, however, take an algorithmic approach.
Whereas our work is proof-theoretic.
Ben-Zvi and Moses \cite{BenZ+Mose:2010,BenZ:2011} take the same approach to coin the \textit{Syncausality} as an extension to \textsf{happens-before} for synchronised computations.
Gonczarowski and Moses \cite{Gonc+Mose:2013} too generalise the classic model to characterise the interactive epistemic state when temporal constraints must be met.
The final work in this thread \cite{Abdu:2005} extends the classic model for reasoning about trust in distributed settings.

Burckhardt \cite{Burc:2014} takes a novel approach to define causal consistency not just in terms of \textsf{happens-before}, but also w.r.t. arbitration order and visibility order.
The gain is a more precise definition of how causality is used to ensure consistency.
In addition to being model theoretic, unlike our work, his approach is not based on explicit causality \cite{Bail+Feke+Ghod+Hell+Stoi:2012}.

One particular motivation for confining the universal knowledge of a world of events to microcosms is scalability.
Systems that reduce the overhead of maintaining scalable causal consistency in wide-area replicated key-value stores include Orbe \cite{Du+Elni+Roy+Zwae:2013}, COPS \cite{Lloy+Free+Kami+Ande:2011}, Eiger \cite{Lloy+Free+Kami+Ande:2013}, and ChainReaction \cite{Alme+Leit+Rodr:2013}.
COPS, in particular, defines \textit{causal+ consistency}, which extends causal consistency with convergent conflict handling.
This ensures that replicas that see concurrent updates will be updated in a consistent fashion.
The systems mentioned above can incur significant overhead (in computation, storage, network load, and latency) to maintain causal consistency in scalable fashion.
Du et. al \cite{Du+Iorg+Roy+Zwae:2014} explain the performance overhead of causal consistency vs. eventual consistency.
They introduce a protocol to reduce this overhead at the cost of degrading the quality-of-service (offered to the client) by significantly increasing data staleness.


\section{Conclusion and Future Work} \label{Sect:Conclusion}

To the best of our knowledge, this is the first proof-theoretic modelling of causality in distributed systems, with special emphasis on partiality of causal knowledge.
In our model, a device has strictly less causal information than a holistic causality store (Lemma~\ref{Lemm:WoE.More.Accurate.Micro}).
We offer rules for deducing causal information both when a device is online and offline (Definitions~\ref{Defn:Micro.Decision} and \ref{Defn:Offline.Rules}).
We prove properties of our deductions, which are both theoretically attractive and practically valuable (Theorems~\ref{Theo:Rules.Comp}, \ref{Theo:Micro.Rules.Const}, Corollary~\ref{Corr:Mu.Hats}, and Lemmata~\ref{Lemm:Offline.Not.Refute.Online} and \ref{Lemm:Offline.Not.Disagree.Online}).
We refute a causality folklore using a mechanical proof (Lemma~\ref{Lemm:Sample.Refutation.1}).
We define two notions of bisimilarity (Definitions~\ref{Defn:For.Bisimilarity} and \ref{Defn:Back.Bisimilarity}) to prove that the order of addition or removal of causal data is irrelevant for bisimilar devices (Theorems~\ref{Theo:For.Order.Irrelevant} and \ref{Theo:Back.Order.Irrelevant}, respectively).

There are two immediate improvements to our model that form interesting future work.
The first is the study of how to retain (\ref{Misc:M4}) whilst still not disallowing arrival of new information (like Fig.~\ref{Fig:NVCs}).
The second is getting forward bisimilarity (and, therefore, backward bisimilarity) to also consider evolution from one microcosm to another by updates (as well as additions).

Our modelling does not take it into consideration that information about concurrent events might arrive not at the same time.
That lag makes a device observe an internal ordering for concurrent events.
The interplay between the concurrency and the internal order becomes more interesting when relaying the concurrency to the next device in the vicinity.
Studying that interplay is future work.
We anticipate that a new set of proof systems will be required, their status w.r.t. the ones in this paper also requires dedicated study.
Another related future work is to take arbitration and visibility into account.

The ability to reason about partial causal information suggests positive interaction with causal+ consistency:
replicas that are actually causal but for which the causality is not known yet will remain consistently updated as the known causality increases (i.e., updates do not have to be redone as knowledge increases).
This is an important property of causal+ consistency that can be a useful model to have together with the deduction systems introduced in this paper.
Future work will reveal how the ability to deduce causality can increase the efficiency of COPS (and its counterparts) by reducing the overhead.

\paragraph*{\textbf{Acknowledgements}}

This work was partially funded by the \href{syncfree.lip6.fr}{SyncFree} project in the European Seventh Framework Programme under Grant Agreement 609551 and by the Erasmus Mundus Joint Doctorate Programme under Grant Agreement 2012-0030.
Our special thanks to the \href{syncfree.lip6.fr}{SyncFree} peers for their prolific comments on the early versions of this work.
We would like to also thank the anonymous referees for their constructive discussion over the ICE forum.

\bibliographystyle{eptcs}
\bibliography{WoE}

\newcommand{\SC}[1]{#1 SC}\newcommand{\OOPSLA}[1]{#1
  OOPSLA}\newcommand{\ECOOP}[1]{#1 ECOOP}\newcommand{\POPL}[1]{#1 POPL}
\begin{thebibliography}{10}
\providecommand{\bibitemdeclare}[2]{}
\providecommand{\surnamestart}{}
\providecommand{\surnameend}{}
\providecommand{\urlprefix}{Available at }
\providecommand{\url}[1]{\texttt{#1}}
\providecommand{\href}[2]{\texttt{#2}}
\providecommand{\urlalt}[2]{\href{#1}{#2}}
\providecommand{\doi}[1]{doi:\urlalt{http://dx.doi.org/#1}{#1}}
\providecommand{\bibinfo}[2]{#2}

\bibitemdeclare{phdthesis}{Abdu:2005}
\bibitem{Abdu:2005}
\bibinfo{author}{A.~\surnamestart Abdul-Rahman\surnameend}
  (\bibinfo{year}{2005}): \emph{\bibinfo{title}{{A Framework for Decentralised
  Trust Reasoning}}}.
\newblock Ph.D. thesis, \bibinfo{school}{U. London}.

\bibitemdeclare{inproceedings}{Alme+Leit+Rodr:2013}
\bibitem{Alme+Leit+Rodr:2013}
\bibinfo{author}{S.~\surnamestart Almeida\surnameend},
  \bibinfo{author}{J.~\surnamestart Leit{\~{a}}o\surnameend} \&
  \bibinfo{author}{L.~E.~T. \surnamestart Rodrigues\surnameend}
  (\bibinfo{year}{2013}): \emph{\bibinfo{title}{{ChainReaction: A Causal+
  Consistent Datastore Based on Chain Replication}}}.
\newblock In \bibinfo{editor}{Z.~\surnamestart Hanz{\'{a}}lek\surnameend},
  \bibinfo{editor}{H.~\surnamestart H{\"{a}}rtig\surnameend},
  \bibinfo{editor}{M.~\surnamestart Castro\surnameend} \&
  \bibinfo{editor}{M.~F. \surnamestart Kaashoek\surnameend}, editors: {\sl
  \bibinfo{booktitle}{$8^\mathit{th}$ EuroSys}}, \bibinfo{publisher}{{ACM}},
  pp. \bibinfo{pages}{85--98}, \doi{10.1145/2465351.2465361}.

\bibitemdeclare{inproceedings}{Bail+Feke+Ghod+Hell+Stoi:2012}
\bibitem{Bail+Feke+Ghod+Hell+Stoi:2012}
\bibinfo{author}{P.~\surnamestart Bailis\surnameend} et~al.
  (\bibinfo{year}{2012}): \emph{\bibinfo{title}{{The Potential Dangers of
  Causal Consistency and an Explicit Solution}}}.
\newblock In \bibinfo{editor}{M.~J. \surnamestart Carey\surnameend} \&
  \bibinfo{editor}{S.~\surnamestart Hand\surnameend}, editors: {\sl
  \bibinfo{booktitle}{{$3^{rd}$ SOCC}}}, \bibinfo{publisher}{{ACM}}, pp.
  \bibinfo{pages}{22--1--22--7}, \doi{10.1145/2391229.2391251}.

\bibitemdeclare{phdthesis}{BenZ:2011}
\bibitem{BenZ:2011}
\bibinfo{author}{I.~\surnamestart Ben-Zvi\surnameend} (\bibinfo{year}{2010}):
  \emph{\bibinfo{title}{{Causality, Knowledge and Coordination in Distributed
  Systems}}}.
\newblock Ph.D. thesis, \bibinfo{school}{Technion}.

\bibitemdeclare{inproceedings}{BenZ+Mose:2010}
\bibitem{BenZ+Mose:2010}
\bibinfo{author}{I.~\surnamestart Ben{-}Zvi\surnameend} \&
  \bibinfo{author}{Y.~\surnamestart Moses\surnameend} (\bibinfo{year}{2010}):
  \emph{\bibinfo{title}{{Beyond Lamport's \emph{Happened-Before}: On the Role
  of Time Bounds in Synchronous Systems}}}.
\newblock In \bibinfo{editor}{N.~A. \surnamestart Lynch\surnameend} \&
  \bibinfo{editor}{A.~A. \surnamestart Shvartsman\surnameend}, editors: {\sl
  \bibinfo{booktitle}{$24\mathit{th}$ DISC}}, {\sl \bibinfo{series}{LNCS}}
  \bibinfo{volume}{6343}, \bibinfo{publisher}{Springer}, pp.
  \bibinfo{pages}{421--436}, \doi{10.1007/978-3-642-15763-9\_42}.

\bibitemdeclare{article}{Burc:2014}
\bibitem{Burc:2014}
\bibinfo{author}{S.~\surnamestart Burckhardt\surnameend}
  (\bibinfo{year}{2014}): \emph{\bibinfo{title}{{Principles of Eventual
  Consistency}}}.
\newblock {\sl \bibinfo{journal}{FTPL}}
  \bibinfo{volume}{1}(\bibinfo{number}{1-2}), pp. \bibinfo{pages}{1--150},
  \doi{10.1561/2500000011}.

\bibitemdeclare{article}{Char:1991}
\bibitem{Char:1991}
\bibinfo{author}{B.~\surnamestart Charron{-}Bost\surnameend}
  (\bibinfo{year}{1991}): \emph{\bibinfo{title}{{Concerning the Size of Logical
  Clocks in Distributed Systems}}}.
\newblock {\sl \bibinfo{journal}{Inf. Proc. Lett.}}
  \bibinfo{volume}{39}(\bibinfo{number}{1}), pp. \bibinfo{pages}{11--16},
  \doi{10.1016/0020-0190(91)90055-M}.

\bibitemdeclare{inproceedings}{Deca+Hast+Jamp+Kaku+Laks+Pilc+Siva+Voss+Voge:2007}
\bibitem{Deca+Hast+Jamp+Kaku+Laks+Pilc+Siva+Voss+Voge:2007}
\bibinfo{author}{G.~\surnamestart DeCandia\surnameend} et~al.
  (\bibinfo{year}{2007}): \emph{\bibinfo{title}{{Dynamo: Amazon's Highly
  Available Key-Value Store}}}.
\newblock In: {\sl \bibinfo{booktitle}{$21^\mathit{st}$ SOSP}}, pp.
  \bibinfo{pages}{205--220}, \doi{10.1145/1294261.1294281}.

\bibitemdeclare{inproceedings}{Du+Elni+Roy+Zwae:2013}
\bibitem{Du+Elni+Roy+Zwae:2013}
\bibinfo{author}{J.~\surnamestart Du\surnameend} et~al. (\bibinfo{year}{2013}):
  \emph{\bibinfo{title}{{Orbe: Scalable Causal Consistency using Dependency
  Matrices and Physical Clocks}}}.
\newblock In \bibinfo{editor}{G.~M. \surnamestart Lohman\surnameend}, editor:
  {\sl \bibinfo{booktitle}{SOCC}}, \bibinfo{publisher}{{ACM}}, pp.
  \bibinfo{pages}{11:1--11:14}, \doi{10.1145/2523616.2523628}.

\bibitemdeclare{inproceedings}{Du+Iorg+Roy+Zwae:2014}
\bibitem{Du+Iorg+Roy+Zwae:2014}
\bibinfo{author}{J.~\surnamestart Du\surnameend} et~al. (\bibinfo{year}{2014}):
  \emph{\bibinfo{title}{{Closing the Performance Gap between Causal Consistency
  and Eventual Consistency}}}.
\newblock In: {\sl \bibinfo{booktitle}{$1^\mathit{st}$ PaPEC}},
  \bibinfo{volume}{EPFL-CONF-198281}, \bibinfo{publisher}{ACM}.

\bibitemdeclare{incollection}{Fagi+Halp+Mose+Vard:2003}
\bibitem{Fagi+Halp+Mose+Vard:2003}
\bibinfo{author}{R.~\surnamestart Fagin\surnameend} et~al.
  (\bibinfo{year}{2003}): \emph{\bibinfo{title}{{Common Knowledge Revisited}}}.
\newblock In \bibinfo{editor}{V.~F. \surnamestart Hendricks\surnameend},
  \bibinfo{editor}{K.~F. \surnamestart J{\o}rgensen\surnameend} \&
  \bibinfo{editor}{S.~A. \surnamestart Pedersen\surnameend}, editors: {\sl
  \bibinfo{booktitle}{{Knowledge Contributors}}}, {\sl
  \bibinfo{series}{{Synthese Library}}} \bibinfo{volume}{322},
  \bibinfo{publisher}{Springer Netherlands}, pp. \bibinfo{pages}{87--104},
  \doi{10.1007/978-94-007-1001-6\_5}.

\bibitemdeclare{inproceedings}{Gilb+Lync:2002}
\bibitem{Gilb+Lync:2002}
\bibinfo{author}{S.~\surnamestart Gilbert\surnameend} \&
  \bibinfo{author}{N.~\surnamestart Lynch\surnameend} (\bibinfo{year}{2002}):
  \emph{\bibinfo{title}{{Brewer's Conjecture and the Feasibility of Consistent
  Available Partition-Tolerant Web Services}}}.
\newblock \bibinfo{volume}{33}, pp. \bibinfo{pages}{51--59},
  \doi{10.1145/564585.564601}.

\bibitemdeclare{inproceedings}{Gonc+Mose:2013}
\bibitem{Gonc+Mose:2013}
\bibinfo{author}{Y.~A. \surnamestart Gonczarowski\surnameend} \&
  \bibinfo{author}{Y.~\surnamestart Moses\surnameend} (\bibinfo{year}{2013}):
  \emph{\bibinfo{title}{Timely Common Knowledge}}.
\newblock In \bibinfo{editor}{B.~C. \surnamestart Schipper\surnameend}, editor:
  {\sl \bibinfo{booktitle}{$14^\mathit{th}$ TARK}}.

\bibitemdeclare{inproceedings}{Gots+Yang+Ferr+Naja+Shap:2016}
\bibitem{Gots+Yang+Ferr+Naja+Shap:2016}
\bibinfo{author}{A.~\surnamestart Gotsman\surnameend} et~al.
  (\bibinfo{year}{2016}): \emph{\bibinfo{title}{{'Cause I'm Strong Enough:
  Reasoning about Consistency Choices in Distributed Systems}}}.
\newblock In \bibinfo{editor}{R.~\surnamestart Bod{\'{\i}}k\surnameend} \&
  \bibinfo{editor}{R.~\surnamestart Majumdar\surnameend}, editors: {\sl
  \bibinfo{booktitle}{\POPL{$43^\mathit{rd}$}}}, \bibinfo{publisher}{{ACM}},
  pp. \bibinfo{pages}{371--384}, \doi{10.1145/2837614.2837625}.

\bibitemdeclare{article}{Halp+Mose:1990}
\bibitem{Halp+Mose:1990}
\bibinfo{author}{J.~Y. \surnamestart Halpern\surnameend} \&
  \bibinfo{author}{Y.~\surnamestart Moses\surnameend} (\bibinfo{year}{1990}):
  \emph{\bibinfo{title}{{Knowledge and Common Knowledge in a Distributed
  Environment}}}.
\newblock {\sl \bibinfo{journal}{JACM}}
  \bibinfo{volume}{37}(\bibinfo{number}{3}), pp. \bibinfo{pages}{549--587},
  \doi{10.1145/79147.79161}.

\bibitemdeclare{article}{Lamp:1978}
\bibitem{Lamp:1978}
\bibinfo{author}{L.~\surnamestart Lamport\surnameend} (\bibinfo{year}{1978}):
  \emph{\bibinfo{title}{{Time, Clocks, and the Ordering of Events in a
  Distributed System}}}.
\newblock {\sl \bibinfo{journal}{Commun. ACM}}
  \bibinfo{volume}{21}(\bibinfo{number}{7}), pp. \bibinfo{pages}{558--565},
  \doi{10.1145/359545.359563}.

\bibitemdeclare{inproceedings}{Lloy+Free+Kami+Ande:2011}
\bibitem{Lloy+Free+Kami+Ande:2011}
\bibinfo{author}{W.~\surnamestart Lloyd\surnameend} et~al.
  (\bibinfo{year}{2011}): \emph{\bibinfo{title}{{Don't Settle for Eventual:
  Scalable Causal Consistency for Wide-Area Storage with COPS}}}.
\newblock In: {\sl \bibinfo{booktitle}{$23^\mathit{rd}$ SOSP}},
  \bibinfo{publisher}{ACM}, \bibinfo{address}{New York, NY, USA}, pp.
  \bibinfo{pages}{401--416}, \doi{10.1145/2043556.2043593}.

\bibitemdeclare{inproceedings}{Lloy+Free+Kami+Ande:2013}
\bibitem{Lloy+Free+Kami+Ande:2013}
\bibinfo{author}{W.~\surnamestart Lloyd\surnameend} et~al.
  (\bibinfo{year}{2013}): \emph{\bibinfo{title}{{Stronger Semantics for
  Low-Latency Geo-Replicated Storage}}}.
\newblock In \bibinfo{editor}{N.~\surnamestart Feamster\surnameend} \&
  \bibinfo{editor}{J.~C. \surnamestart Mogul\surnameend}, editors: {\sl
  \bibinfo{booktitle}{$10^\mathit{th}$ NSDI}}, \bibinfo{publisher}{USENIX}, pp.
  \bibinfo{pages}{313--328}.

\bibitemdeclare{article}{Schn:1990}
\bibitem{Schn:1990}
\bibinfo{author}{F.~B \surnamestart Schneider\surnameend}
  (\bibinfo{year}{1990}): \emph{\bibinfo{title}{{Implementing Fault-Tolerant
  Services using the State Machine Approach: A Tutorial}}}.
\newblock {\sl \bibinfo{journal}{ACM CSUR}}
  \bibinfo{volume}{22}(\bibinfo{number}{4}), pp. \bibinfo{pages}{299--319},
  \doi{10.1145/98163.98167}.

\bibitemdeclare{article}{Schw+Matt:1994}
\bibitem{Schw+Matt:1994}
\bibinfo{author}{R.~\surnamestart Schwarz\surnameend} \&
  \bibinfo{author}{F.~\surnamestart Mattern\surnameend} (\bibinfo{year}{1994}):
  \emph{\bibinfo{title}{{Detecting Causal Relationships in Distributed
  Computations: In Search of the Holy Grail}}}.
\newblock {\sl \bibinfo{journal}{Dist. Comp.}}
  \bibinfo{volume}{7}(\bibinfo{number}{3}), pp. \bibinfo{pages}{149--174},
  \doi{10.1007/BF02277859}.

\bibitemdeclare{inproceedings}{Serg+Baqu+Gonc+Preg+Font:2014}
\bibitem{Serg+Baqu+Gonc+Preg+Font:2014}
\bibinfo{author}{P.~\surnamestart S{\'{e}}rgio~Almeida\surnameend} et~al.
  (\bibinfo{year}{2014}): \emph{\bibinfo{title}{{Scalable and Accurate
  Causality Tracking for Eventually Consistent Stores}}}.
\newblock In: {\sl \bibinfo{booktitle}{$14^\mathit{th}$ IFIP DAIS}}, pp.
  \bibinfo{pages}{67--81}, \doi{10.1007/978-3-662-43352-2\_6}.

\bibitemdeclare{book}{Yap:1998}
\bibitem{Yap:1998}
\bibinfo{author}{C.~K. \surnamestart Yap\surnameend} (\bibinfo{year}{1998}):
  \emph{\bibinfo{title}{{Theory of Complexity Classes}}}.
\newblock \bibinfo{volume}{1}.
\newblock
  \bibinfo{note}{\href{https://cs.nyu.edu/yap/book/complexity/}{\texttt{https://cs.nyu.edu/yap/book/complexity/}}}.

\end{thebibliography}

\end{document}